\documentclass[10pt,journal,compsoc]{IEEEtran}

\ifCLASSINFOpdf
\usepackage[pdftex]{graphicx}
\DeclareGraphicsExtensions{.pdf,.jpeg,.png}
\else
\fi

\usepackage{algorithm}
\usepackage{algorithmic}

\usepackage{subfigure}

\usepackage{amsmath}
\usepackage{amssymb}
\usepackage{bold-extra}
\usepackage{color}
\usepackage{graphicx}
\usepackage{amsthm}
\usepackage{csquotes}
\usepackage{tikz}
\usepackage{url}
\usepackage[T1]{fontenc}
\usepackage[utf8]{inputenc}
\usetikzlibrary{calc}
\usetikzlibrary{positioning}
\usepackage[absolute,overlay]{textpos} % for svn rev
\usepackage{latexsym}
\usepackage{csquotes}
\usepackage{color}
\usepackage{makecell}
\usepackage{diagbox}
\usepackage{multirow}
\newcommand{\todo}[1]{\textcolor[rgb]{0.00,0.00,0.00}{#1}}

\newcommand{\forget}[1]{}

\newtheorem{lemma}{Lemma}
\newtheorem{define}{Definition}
\newtheorem{theorem}{Theorem}

\newtheorem{corollary}{Corollary}

\newcommand{\njobs}{\eta}

\newcommand{\Tau}{\mathcal{T}}

\newcommand{\critical}{\mathcal{L}}

\newcommand{\resource}{\ell}
\newcommand{\resourceset}{\Theta}
\newcommand{\interfer}{\mathcal{I}}
\newcommand{\sumC}{\mathcal{C}}
\newcommand{\maxL}{\mathcal{L}}

\newcommand{\reduceblack}{\vspace{-2.5mm}}
\newcommand{\reduceblackhalf}{\vspace{-1.25mm}}

\ifCLASSOPTIONcompsoc
\usepackage[nocompress]{cite}
\else
\usepackage{cite}
\fi

\begin{document}
	
	%
	% paper title
	% Titles are generally capitalized except for words such as a, an, and, as,
	% at, but, by, for, in, nor, of, on, or, the, to and up, which are usually
	% not capitalized unless they are the first or last word of the title.
	% Linebreaks \\ can be used within to get better formatting as desired.
	% Do not put math or special symbols in the title.
	\title{On the Analysis of Parallel Real-Time Tasks with Spin Locks}

	\author{Xu Jiang$^{1,2}$,  Nan Guan$^{2*}$, He Du$^{1,2}$, Weichen Liu$^{3}$, Wang Yi$^{4}$
		~~\\
		~~\\
		$^1$ Northeastern University, China\\
		$^2$ The Hong Kong Polytechnic University, Hong Kong\\
		$^3$ Nanyang Technological University, Singapore\\
		$^4$ Uppsala University, Sweden\\
		%\IEEEcompsocitemizethanks{
		%\IEEEcompsocthanksitem
		%}
		\thanks{*Corresponding author: Nan Guan}
	}

	% The paper headers
	\markboth{Journal of \LaTeX\ Class Files,~Vol.~XX, No.~XX, XX~201X}%
	{Shell \MakeLowercase{\textit{et al.}}: Bare Advanced Demo of IEEEtran.cls for IEEE Computer Society Journals}
	% The only time the second header will appear is for the odd numbered pages
	% after the title page when using the twoside option.
	% 
	% *** Note that you probably will NOT want to include the author's ***
	% *** name in the headers of peer review papers.                   ***
	% You can use \ifCLASSOPTIONpeerreview for conditional compilation here if
	% you desire.

	% The publisher's ID mark at the bottom of the page is less important with
	% Computer Society journal papers as those publications place the marks
	% outside of the main text columns and, therefore, unlike regular IEEE
	% journals, the available text space is not reduced by their presence.
	% If you want to put a publisher's ID mark on the page you can do it like
	% this:
	%\IEEEpubid{0000--0000/00\$00.00~\copyright~2014 IEEE}
	% or like this to get the Computer Society new two part style.
	%\IEEEpubid{\makebox[\columnwidth]{\hfill 0000--0000/00/\$00.00~\copyright~2014 IEEE}%
	%\hspace{\columnsep}\makebox[\columnwidth]{Published by the IEEE Computer Society\hfill}}
	% Remember, if you use this you must call \IEEEpubidadjcol in the second
	% column for its text to clear the IEEEpubid mark (Computer Society journal
	% papers don't need this extra clearance.)

	% use for special paper notices
	%\IEEEspecialpapernotice{(Invited Paper)}

	% for Computer Society papers, we must declare the abstract and index terms
	% PRIOR to the title within the \IEEEtitleabstractindextext IEEEtran
	% command as these need to go into the title area created by \maketitle.
	% As a general rule, do not put math, special symbols or citations
	% in the abstract or keywords.
	\IEEEtitleabstractindextext{%
		\begin{abstract}
			Locking protocol is an essential component in resource management of real-time systems, which coordinates mutually exclusive accesses to shared resources from different tasks.
Although the design and analysis of locking protocols 
have been intensively studied for \emph{sequential} real-time tasks, there has been little work on this topic for \emph{parallel} real-time tasks. In this paper, we study the analysis of parallel real-time tasks using spin locks
to protect accesses to shared resources in three commonly used request serving orders (unordered, FIFO-order and priority-order). A remarkable feature making our analysis method more accurate is to systematically analyze the blocking time which may delay a task's finishing time, where the impact to the total workload and the longest path length is \emph{jointly} considered, rather than analyzing them \emph{separately} and counting all blocking time as the workload that delays a task's finishing time, as commonly assumed in the state-of-the-art. 

%Experiments with both randomly synthesized tasks and workload generated according to realistic OpenMP programs show that our techniques consistently outperform the state-of-the-art under different experiment settings. \todo{finer-grained classification of different blocking types}
%, rather than analyzing them \emph{separately} and counting all of them in the workload that delays a task's finish in the state-of-the-art.

		\end{abstract}
		% Note that keywords are not normally used for peerreview papers.
		}
	
	% make the title area
	\maketitle

	% To allow for easy dual compilation without having to reenter the
	% abstract/keywords data, the \IEEEtitleabstractindextext text will
	% not be used in maketitle, but will appear (i.e., to be "transported")
	% here as \IEEEdisplaynontitleabstractindextext when compsoc mode
	% is not selected <OR> if conference mode is selected - because compsoc
	% conference papers position the abstract like regular (non-compsoc)
	% papers do!
	\IEEEdisplaynontitleabstractindextext
	% \IEEEdisplaynontitleabstractindextext has no effect when using
	% compsoc under a non-conference mode.

	% For peer review papers, you can put extra information on the cover
	% page as needed:
	% \ifCLASSOPTIONpeerreview
	% \begin{center} \bfseries EDICS Category: 3-BBND \end{center}
	% \fi
	%
	% For peerreview papers, this IEEEtran command inserts a page break and
	% creates the second title. It will be ignored for other modes.
	\IEEEpeerreviewmaketitle

	% Computer Society journal (but not conference!) papers do something unusual
	% with the very first section heading (almost always called "Introduction").
	% They place it ABOVE the main text! IEEEtran.cls does not automatically do
	% this for you, but you can achieve this effect with the provided
	% \IEEEraisesectionheading{} command. Note the need to keep any \label that
	% is to refer to the section immediately after \section in the above as
	% \IEEEraisesectionheading puts \section within a raised box.

	% The very first letter is a 2 line initial drop letter followed
	% by the rest of the first word in caps (small caps for compsoc).
	% 
	% form to use if the first word consists of a single letter:
	% \IEEEPARstart{A}{demo} file is ....
	% 
	% form to use if you need the single drop letter followed by
	% normal text (unknown if ever used by IEEE):
	% \IEEEPARstart{A}{}demo file is ....
	% 
	% Some journals put the first two words in caps:
	% \IEEEPARstart{T}{his demo} file is ....
	% 
	% Here we have the typical use of a "T" for an initial drop letter
	% and "HIS" in caps to complete the first word.
	%\IEEEPARstart{T}{his} demo file is intended to serve as a ``starter file''
	%for IEEE Computer Society journal papers produced under \LaTeX\ using
	%IEEEtran.cls version 1.8a and later.
	%% You must have at least 2 lines in the paragraph with the drop letter
	%% (should never be an issue)
	%I wish you the best of success.
	%
	%\hfill mds
	% 
	%\hfill September 17, 2014

	\section{Introduction}

Real-time systems are playing a more important role in our daily life as computing is closely integrated to the physical world. Violating timing constraints in such systems may lead to catastrophic consequences such as loss of human life. Therefore, real-time systems must manage resource in a way such that timing correctness can be guaranteed. Locking protocol is an essential component in resource management of real-time systems, which coordinates mutually exclusive accesses to shared physical/logical resources by different tasks. Inappropriate design or incorrect analysis of locking protocols will lead to incorrect system timing behavior, e.g., as in the famous software failure accident in Mars Pathfinder \cite{jones1997really}.

Multi-cores are becoming mainstream hardware platforms for real-time systems, to meet their rapidly increasing requirements in high performance and low power consumption. To fully utilize the processing capacity of multi-cores, software should be parallelized. While locking protocols for \emph{sequential} real-time task systems have been intensively studied in classical real-time scheduling theory \cite{baker1991stack,sha1990priority,brandenburg2010optimality,rajkumar1990real}, there is little work on this topic for \emph{parallel} real-time tasks. 
On the other hand, there has been much work on scheduling algorithms and analysis techniques for parallel real time tasks \cite{li2014analysis,maia2014response,jiang2016decomposition}, where tasks are  assumed to be independent from each other and the locking issue is not considered.
 
Recently, spin locks were studied for parallel real-time
tasks in \cite{dinh2018blocking} where each parallel task is scheduled exclusively on several pre-assigned processors (i.e., by the \emph{federated} scheduling approach \cite{li2014analysis}). 
However, the analysis in \cite{dinh2018blocking} is pessimistic. The contribution of our work is to develop new techniques for the schedulability analysis of real-time parallel tasks with spin locks and significantly improve the analysis precision against the state-of-the-art.

Both \cite{dinh2018blocking} 
and our work only require  knowledge of the total \todo{worst-case execution} time (WCET) $\sumC_i$ and longest path length $\maxL_i$ of each task, but not the exact  
graph structure (the benefits of only using the abstract $\sumC_i$ and $\maxL_i$ information in the analysis will be discussed in Section \ref{ss: remark}).
In \cite{dinh2018blocking}'s analysis, \emph{all} blocking time caused by spin locks is considered to contribute to the workload that delays the finishing time of a parallel task, which is added to $\sumC_i$ and $\maxL_i$ in their worst-case scenarios \emph{separately}. This is quite pessimistic since many blocking time can \emph{not} delay the finishing time of a parallel task due to the parallelism and intra-dependencies. Moreover, the worst-case scenario leading to the maximal increase to $\sumC_i$ is in general different from the worst-case scenario leading to the maximal increase to $\maxL_i$.

To solve these problems, in this work we first develop new schedulability analysis techniques for parallel tasks with spin locks, where the blocking time contributing to the workload that may delay a task's finishing time is systematically defined and analyzed. Further, we develop blocking analysis techniques for three common request serving orders, i.e., unordered, FIFO-order and priority-order, where the impact to $\sumC_i$ and $\maxL_i$ is \emph{jointly} considered thus achieving higher analysis precision.

We conduct experiments to evaluate the precision improvement using our new techniques compared with \cite{dinh2018blocking}, with both randomly generated tasks and workload generated according to realistic OpenMP programs. Experimental results show that our techniques consistently outperform \cite{dinh2018blocking} under different settings.

%\footnote{The the exact  
%	graph structures
%	
%	
%	
%	} (which are often hard or even impossible to obtain in design time). 
%
%
%
%e blocking times and core assignment.
%Note that this schedulability test only requires basic information
%about the tasks including worst case execution time
%(work), critical-path length (span)

%
%As the growing prevalence of multi-core processors, parallel programming models (such as OpenMP \cite{openmp2013openmp}) have drawn increasing attentions in real-time systems, to exploit the intra-task parallelism such that rapidly increasing performance requirements and more cruel real-time constraints can be satisfied. 
%

% \textbf{existing work are all for independent model}
%
% 
%
%In this paper, we analyze response times of DAG tasks, federated, spin locks.  
%
%\textbf{why federated}
%
%\textbf{why spin lock}
%
%\textbf{our contributions}

	\section{Preliminary}\label{s:pre}
\subsection{Task Model}

We consider a task set $\Tau$ consisting of several periodic DAG tasks $\Tau=\{\tau_1, \tau_2,..., \tau_{|\Tau|}\}$ to be executed on $m$ processors. A task $\tau_i$ has a \emph{period} $T_i$, a \emph{relative deadline} $D_i$ and 
 a workload structure modeled by a Directed Acyclic Graph (DAG) $G_i = \langle V_i, E_i \rangle$, where $V_i$ is the set of vertices and $E_i$ is the set of edges in $G_i$. 
 Tasks have \emph{constrained} deadlines, i.e., $D_i \leq T_i$.
  Each vertex $v \in V_i$ is characterized by a worst-case execution time (WCET) $c(v)$. We use $\sumC_i$ to denote the total WCET of all vertices of $\tau_i$: $\sumC_i=\sum_{v\in V_i}c(v)$.
The \emph{utilization}
of task $\tau_i$ is $U_i = \sumC_i/T_i$ and the \emph{density} of task $\tau_i$ is $\Gamma_i=\sumC_i/D_i$. In this paper, we only consider DAG tasks with $\Gamma_i > 1$, as those with $\Gamma_i \leq 1$ can be executed sequentially and handled by existing techniques for sequential real-time tasks.
  
%\guan{unify $\tau_i$ and $G_i$, no need to use two concepts as they are almost identically }
Each edge $(u, v) \in E_i$ represents the precedence relation between vertices $u$ and $v$, where $u$ is a \emph{predecessor} of $v$, and $v$ is a \emph{successor} of $u$. We assume each DAG has a unique head vertex (with no predecessors) and a unique tail vertex (with no successors). This assumption does not limit the expressiveness of our model since one can always add a dummy head/tail vertex to a DAG having multiple entry/exit points. A \emph{complete path} in a DAG task is a sequence of vertices $\pi=\{v_1,v_2,...,v_p\} $, where
the first element $v_1$ is the head vertex of $G_i$, the last element 
$v_p$ is the tail vertex of $G_i$, 
and $v_{j}$ is a predecessor of $v_{j+1}$ for each pair of consecutive elements $v_{j}$ and $v_{j+1}$ in $\pi$. The length of each path $\pi$ is $len(\pi) =  \sum_{v \in \pi} c(v)$. 
%\guan{$len(\pi) = \sum_{v \in \pi} c(v) $}
We use $\critical_i$ to denote the longest length among all paths in $G_i$: $\critical_i = \max_{\pi \in G_i} \{len(\pi) \}$.
Task $\tau_i$ generates a potentially infinite sequence of jobs, which inherit $\tau_i$'s DAG workload structure $G_i$. 
Let $J$ be a job released by $\tau_i$, then we use $r(J)$ to denote $J$'s release time and $f(J)$ to denote $J$'s finish time.
The \emph{absolute deadline} of $J$ is calculated by $r(J)+D_i$.
%Each job must be finished before its absolute deadline calculated by $r_i+D_i$. In this paper, we consider DAG tasks with \emph{implicit deadline}, i.e., we have the minimum time interval between two successive jobs $T_i$ equals to $D_i$. The \emph{utilization} $U_i$ of $\tau_i$ is defined as $U_i = \frac{C_i}{D_i}$. 
%A job $J$ is said to be \emph{pending} at time $t$ if it has been released but not finished by $t$. 
At runtime, we say a vertex (of a job $J$) is \textit{eligible} at some time point if all its predecessors (of the same job $J$) have been finished and thus it can immediately execute if there are available processors. Fig. \ref{f:DAG} shows a DAG task example $\tau_i$ with $7$ vertices, where $\sumC_i = 10$ 
and $\maxL_i = 5$ (the longest path is $\{v_1, v_4, v_6 , v_7\}$ or $\{v_1, v_2, v_6 , v_7\}$). 
%\guan{is pending useful?}
%
%\guan{utilization larger than $1$}
%
%\guan{clarify that only the $\sumC_i$ and $maxL_i$ information are really needed}
\begin{figure}
	\centering
	\includegraphics[width=2.2in]{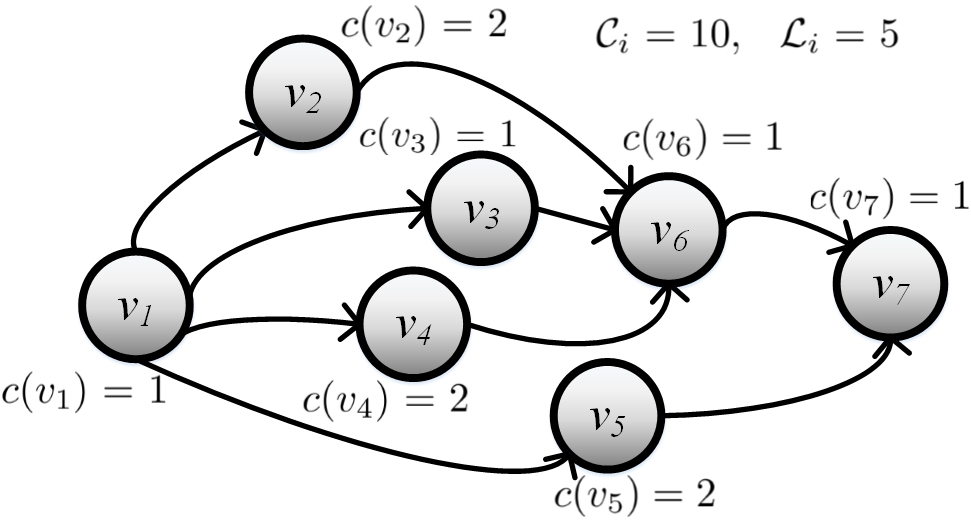}
	\reduceblack
	\caption{An example of a DAG task $\tau_i$.}
	\label{f:DAG}
\end{figure}

\subsection{Resource and Lock Model}
There is a limited set of \emph{serially-reusable shared resources} (called \emph{resources} for short) $\Theta = \{\resource_1, \resource_2, ... , \resource_{|\Theta|}\}$ in the system, \todo{such as I/O ports, network links, message buffers, or other shared data structures.}
Resources are protected by \emph{spin locks}, i.e., 
the program must \emph{acquire}, \emph{hold} and
\emph{release} the lock affiliated to $\resource_q$ before, during and after executing the code segment accessing $\resource_q$.
We assume  
the code segment wrapped by a pair of lock acquisition and lock release does not cross different vertices.
A vertex must execute \emph{non-preemptively} when it is holding a lock.
When a vertex acquires a lock affiliated to $\resource_q$ being {held} by other vertices (either from the same task or from other tasks), the acquiring vertex must spin \emph{non-preemptively} until it successfully obtains the lock,
and we say this vertex is \emph{spinning for} $\resource_q$.

 When multiple vertices are spinning for the same resource at the same time, we consider three kinds of order in which their requests will be served: unordered, FIFO-order and priority-order. In priority-order, each task is assigned a unique priority and all requests from vertices of a same task have the same priority. Note that the priorities are only used to decide the order when requests from different tasks to a resource are served. 

%When multiple vertices are spinning for the same resource at the same time, their requests will be served in the FIFO order.
A vertex may access different shared resources and thus hold different locks. However, we assume the locks are \emph{non-nested}, i.e., a vertex 
never acquires another lock when holding a lock.
We use $\Theta_i$ to denote the set of resources accessed by vertices of task $\tau_i$.
%\guan{do we need to require the vertex holding the resource also execute non-preemptively?? I guess yes. If not, need to revise the above texts...}

%
%Modeling the resource accessing timing behaviors is a tricky problem. 
%In general, the more precise one models the resource accessing timing behavior, potentially more accurate analysis of the system timing behavior is achievable. 
%However, on the other hand, a completely precise model is hard to obtain in practice due to the flexibility and non-determinism of the software behavior. 
%In research of locking protocols for \emph{sequential} real-time tasks,
%the most widely used model is as follows:
%
%\textbf{Model I}.
%
%to provide information of the worst-case 
%accessing time to each resource by a task.
%The only known work for \emph{parallel} real-time tasks  \textbf{Model II}

The worst-case time of each \emph{single} access to $\resource_q$ by task $\tau_i$
	(i.e., the  maximal duration for
	a vertex in $\tau_i$ to hold the lock affiliated to $\resource_q$ {once}) is denoted by $L_{i,q}$, and the worst-case number of accesses to $\resource_q$ by $\tau_i$ is denoted by $N_{i,q}$. Note that a vertex's WCET includes the resource access time. On the contrary, the time spent by a vertex on \emph{spinning for} some resource, called \emph{blocking time} \cite{block2007flexible,wieder2013spin}, is not included in the WCET estimation.

 \begin{table}[!htb]\caption{Notations adopted in this paper. }
 	\centering
 	\scriptsize
 	%\footnotesize
 	%	\tiny
 	\begin{tabular}{ c@{}|c@{}}
 		\hline		
 		Notations & Descriptions \\
 		\hline
 		$\tau_i$& a DAG task \\
 		\hline
 		$G_i$& the workload structure of $\tau_i$ \\
 		\hline
 		$V_i$& the set of vertices in $G_i$ \\
 		\hline
 		$E_i$& the set of edges of $G_i$ \\
 		\hline
 		$c(v)$& WCET of a vertex $v$ \\
 		\hline
 		$\sumC_i$& total WCET of all vertices of $\tau_i$ \\
 		\hline
 		$\pi$& a path \\
 		\hline
 		$\lambda$& a key path \\
 		\hline
 		$len(\pi)$& the total WCET of all vertices on $\pi$ \\
 		\hline
 		$\critical_i$& the longest length among all path of $\tau_i$ \\
 		\hline
 		$\resource_q$& a shared resource\\
 		\hline
 		$N_{i,q}$& number of accesses to $\resource_q$ from $\tau_i$\\
 		\hline
 		$L_{i,q}$& the worst-case time of each single access to $\resource_q$ by $\tau_i$ \\
 		\hline
 		
 		$W_{i}$& working  time of a job of $\tau_i$\\
 		\hline
 		$\Gamma_{i}$& idle time of a job of $\tau_i$\\
 		\hline
 		$B_{i}$& blocking time of a job of $\tau_i$\\
 		\hline 
 		$B_{i}^{\lambda,I}$& intra-task key path blocking time of a job of $\tau_i$\\
 		\hline	
 		$B_{i}^{\lambda,O}$& inter-task key path blocking time of a job of $\tau_i$\\
 		\hline
 		$B_{i}^{\overline{\lambda},I}$& intra-task delay blocking time of a job of $\tau_i$\\
 		\hline
 		$B_{i}^{\overline{\lambda},O}$& inter-task delay blocking time of a job of $\tau_i$\\
 		\hline
 		$B_{i}^{\widetilde{\lambda},I}$& intra-task parallel blocking time of a job of $\tau_i$\\
 		\hline
 		$B_{i}^{\widetilde{\lambda},O}$& inter-task parallel blocking time of a job of $\tau_i$ \\
 		\hline
 		$\interfer_i$& defined in Lemma \ref{t:tightbound}\\
 		\hline
 		$\interfer_{i,q}^I$& defined in (\ref{eq:iiqi}) \\
 		\hline
 		$\interfer_{i,q}^O$& defined in (\ref{eq:iiqo})\\
 		\hline
 		$	\njobs_{i,j}^q$& defined in (\ref{eq:nijq})\\
 		\hline
 		
 		$\varDelta_{i,j}^q$& defined in (\ref{eq: maximumdelaynumber})\\
 		\hline	
 		$R_i$& worst-case response time of $\tau_i$\\
 		\hline
 	\end{tabular}
 	\label{tb:notation}
 \end{table}

\subsection{Scheduling Model}
There are in total $m$ processors in the system, which will be partitioned into several subsets and each subset
is assigned to a task. We use $m_i$ to denote the number of processors assigned to task $\tau_i$. At runtime, $\tau_i$ is scheduled 
by a \emph{work-conserving scheduling} algorithm \cite{li2014analysis}
exclusively on these $m_i$ processors. 
Note that although a task $\tau_i$ executes exclusively on its own $m_i$ processors, its timing behavior is still interfered by other tasks due to the contention on the shared resources. The response time $R(J)$ of a job $J$ is  $R(J) = f(J) - r(J)$, and the \emph{worst-case response time} (\todo{WCRT}) $R_i$ of task $\tau_i$ is the maximum $R(J)$ among all its released jobs $J$. Task $\tau_i$ is \emph{schedulable}
if $R_i \leq D_i$. The problem to solve in this paper is how to partition the $m$ processors to each task such that it is guaranteed to be schedulable.

%and the whole system is  \emph{schedulable} if all the tasks are schedulable.

\subsection{Remark} \label{ss: remark}

The analysis techniques of this paper only require the knowledge of  $\sumC_i$ and $\maxL_i$ of each task $\tau_i$, as well as $N_{i,q}$ and $L_{i,q}$ for each pair of task $\tau_i$ and resource $\resource_q$.
It is not required to know the exact graph structure
of the task, neither the exact distribution of the resource access requests within the task.
This makes our analysis techniques general, in the sense that they are directly applicable to more expressive models, e.g., the conditional DAG model, as long as 
we still can obtain the $\sumC_i$, $\maxL_i$, $N_{i,q}$, $L_{i,q}$ information. Moreover, as pointed out by \cite{dinh2018blocking}, parallel programs are often data-dependent and their internal graph structures usually can only be unfolded at run time, so the exact graph structure of a parallel task can vary from one release to the next. Therefore, the analysis techniques using abstract 
information are more practical than those relying on 
exact graph structure information. 

Of course if one can model the resource access behavior in a more 
detailed manner, e.g., giving the exact worst-case duration of each access and the information about which resource is accessed by which part of the task at which time point, 
it will certainly lead to more precise results in general. However, in practice it is not always possible to model realistic systems with those detailed information due to the flexibility and non-determinism of software behavior. Study on finer-grained resource access models
and the corresponding analysis techniques is left as our future work.

\todo{It is necessary to mention that the main scope of this paper is to present blocking and schedulability analysis techniques when scheduling DAG tasks with spin locks. We do not make any constraint to the scheduler that each paralleled program is scheduled with, as long as the work-conserving is satisfied (e.g., EDF). A limitation in this paper is that we assume locks to be non-nested. In fact, the analysis of nested locks is more complicated even for sequential tasks, and this problem is still vastly open \cite{brandenburg2011scheduling}. However, when nested-locks are used in practice, we can adopt some techniques such as group locks \cite{brandenburg2011scheduling} to transform nested-locks into independent locks such that techniques presented in this paper are still applicable.}

%There may also be cases where the exact graph structure and 
%more concrete information about the resource access, e.g., 
%which resource is accessed by which part of the task at which time point, can be obtained. 
%In this case, one may utilize such concrete information and modify sequential locking protocols to perform more precise analysis (again, with both inter-task and intra-task blocking considered and new blocking analysis techniques developed for some corresponding parallel task schedulabiliy tests), which is out of the scope of this paper. 

%\todo{Of course if one can model the resource access behavior in a more 
%detailed manner, e.g., giving the exact worst-case duration of each access, 
%it will certainly lead to more precise results in general. However, in practice it is not always possible to model realistic systems with those detailed information due to the flexibility and non-determinism of software behavior. Study on finer-grained resource access models
%and the corresponding analysis techniques is left as our future work.} 
\begin{figure}[htb]
	\centering
	\subfigure[An example of a DAG task.]{\includegraphics[width=1.8in]{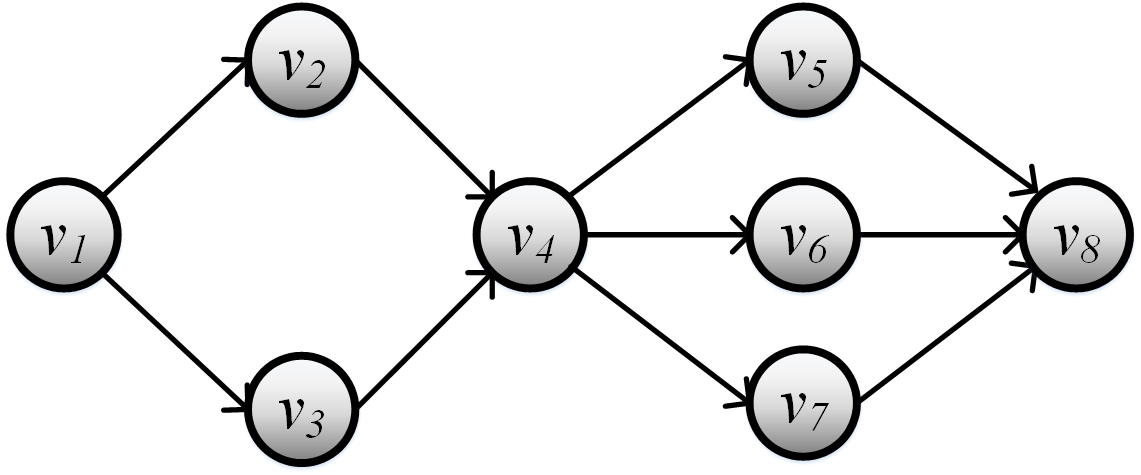}}
	\subfigure[A possible sequence.]{\includegraphics[width=1.7in]{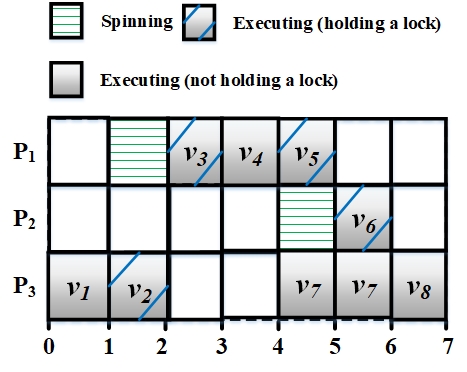}}
	\caption{An example of blocking behavior of a DAG job.}
	\label{fig:diff}
\end{figure}

\section{Discussion of existing techniques } \label{s: sequentialdiscussion}
%Intuitively, the federated scheduling of DAG tasks is similar with the clustered scheduling of sequential tasks where each vertex in a DAG task is regarded as an independent sequential task, e.g., OMLP \cite{brandenburg2010optimality}.
There has been significant work of locking protocols and blocking analysis for sequential tasks (see Section \ref{s:related} for more details). However, it is not a proper choice to directly apply blocking analysis techniques for sequential tasks on DAG tasks. 

% the impact of resource accessing contentions on the schedulability of a DAG job depends on the holistic blocking occurring within it rather than the blocking time that each individual vertex is suffered, and 

First, the definition of \emph{{blocking}} for DAG tasks is different than that under sequential tasks. Under sequential task models, the blocking time of each task is analyzed individually, and the exact definition of blocking time as well as the blocking analysis techniques are developed according to some particular schedulability tests \cite{block2007flexible,brandenburg2010optimality,brandenburg2013omlp} where the blocking time can be accounted in. The main object of locking protocols (with blocking analysis) is to bound such maximum blocking (e.g., the priority inversion blocking \cite{brandenburg2011scheduling,brandenburg2013omlp}) to an individual task. However, this is not the case for DAG tasks where the schedulability analysis object is the whole DAG task. For example, the DAG task in Figure \ref{fig:diff}.(a) has 8 vertices where $c(v_7)=2$ and each of other vertices has a WCET of 1. \todo{$v_2$, $v_3$, $v_5$ and $v_6$ need to access a same shared resource $\resource_q$ for 1 time unit, while the remaining vertices do not need to access any share resource.} A possible execution sequence of a job of $\tau_i$ is shown in Figure \ref{fig:diff}.(b) where $P_1$, $P_2$ and $P_3$ denote the processors. It can be observed that $v_2$ can not be blocked by $v_3$ if $v_2$ blocks $v_3$ in a DAG job (which is also the case for sequential tasks but each vertex must be analyzed individually in a worst-case blocking scenario). Moreover, although $v_6$ is blocked by $v_5$, the finishing time of the DAG job is not delayed. The reason is that the impact of blocking time on the schedualbility of a DAG job is actually reflected by its impact on the
progress of a particular path, i.e., $\{v_1, v_3, v_4, v_7, v_8\}$ in Figure \ref{fig:diff}.(b). These are quite different than that under sequential task models where the blocking time of each task is analyzed individually. To develop blocking analysis for DAG tasks, we first need to systematically define the notion of blocking and analyze which blocking should be accounted according to its influence on the timing behavior of a DAG task. 

Second, as discussed in Section \ref{ss: remark}, the exact distribution of resource access requests is not known under the model considered in this paper. Therefore, it is impossible to directly apply blocking analysis techniques for sequential tasks on the task model considered in this paper. There may also be cases where the exact graph structure and more concrete information about the resource access. In this case, one may utilize such concrete information and use sequential locking protocols to perform blocking analysis. We will evaluate the performance when directly applying OMLP and its associated blocking analysis techniques on DAG tasks in Section \ref{s:evaluate} to validate the problems discussed in this section. 
%
%
%\todo{
%\begin{itemize}
%\item model
%\item definition of blocking 
%\item cluster implementation
%\end{itemize}
%}
%

%Our ultimate goal is to assign the $m$ processors 
%
%
%\guan{$m_i$ should be defined here}
%We adopt \emph{federated scheduling} to schedule all DAG tasks. Under federated scheduling \cite{li2014analysis}, each heavy task $\tau_i$ ($U_i>1$) is assigned several dedicated processors and exclusively executes on them while all light tasks ($U\leq1$) share the remaining processors together as if they are sequential tasks. In this paper, we consider the task set \emph{only} consisting of heavy tasks. The problem is to decide how many processors should be assigned to each heavy task such that it can always guarantee its deadline when it elusively executes on them. 

%Our processor allocation is based on the response time analysis. The response time $R_i$ of $J_i$ is calculated by $R_i=f_i-r_i$ and the sufficient condition for a task $\tau_i$ to be schedulable on a given number of processors is that $\forall J_i,~~R_i\leq D_i$. 

	\section{Preparation}\label{s:pre}

In this section, we introduce some useful concepts and present schedulability analysis techniques for parallel tasks that are applicable irrelevant of the locking protocols and request serving orders. Then in the next section we will apply these results to 
develop specific blocking analysis techniques for  
 unordered, FIFO- and priority- request serving orders, respectively.

%\guan{to revise, probably to the end of theorem 1}
%In this section, we present a general framework for analyzing the response time of a DAG task scheduled on a given number of processors. In particular, this framework is not concerned with the specific resource model being considered.

When we say a vertex is \emph{executing}, it may be either holding or not holding a lock.
We say a processor is \emph{busy} if some vertex is executing or spinning on this processor, and say a processor is \emph{busy with a vertex $v$} if 
vertex $v$ is executing or spinning on this processor.
A processor is said to be \emph{idle} if it is not busy.

Let $J_{i}$ denote an arbitrary job of $\tau_i$, which is released at $r(J_i)$ and finished at $f(J_i)$. 
The total amount of time spent on $m_i$ processors
assigned to $\tau_i$
during $[r(J_i), f(J_i))$ is $m_i \cdot (f(J_i) - r(J_i) )$, which can be divided into three disjoint parts $m_i \cdot (f(J_i) - r(J_i) ) = B_i + W_i + \Gamma_i$:
\begin{itemize}
	\item \textbf{Blocking {Time} $B_{i}$}: the cumulative length of time on $m_i$ processors spent on spinning.
	
	\item \textbf{Working Time $W_{i}$}: the cumulative length of time on $m_i$ processors spent on executing workload of $J_{i}$ (either holding a lock or not).
	
	\item \textbf{Idle Time $\Gamma_{i}$}: the cumulative length of time on $m_i$ processors being idle. 	
\end{itemize}
%Note that 

\begin{figure}[htb]
	\centering
	\includegraphics[width=2.2in]{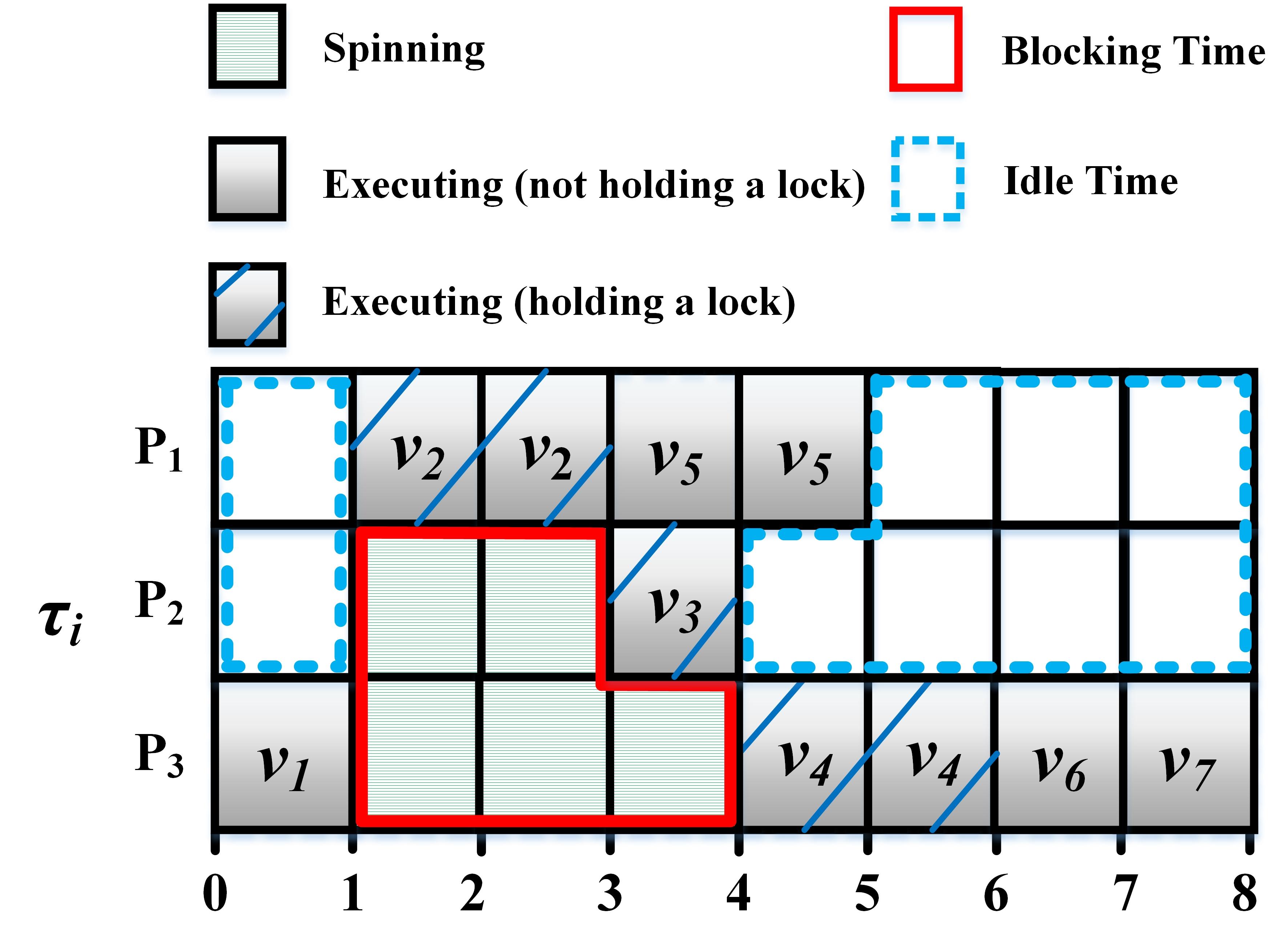}
	\caption{Illustration of different times.}
	\label{f:timetype}
\end{figure}

Fig. \ref{f:timetype} shows a possible scheduling sequence of a job of the task
in Fig. \ref{f:DAG} on $3$ processors, with release time $0$ and finish time $8$.  \todo{Suppose that $v_2$, $v_3$, $v_4$ need to access the same shared resource $\resource_q$ for 1 time unit, while the remaining vertices do not need to access any share resource.}
The blocking time is $5$ (the area wrapped by red solid lines), the idle time is $9$ (the area wrapped by blue dash lines), and the working time
is $10$ (the remaining area between $[0,8)$
on all the $3$ processors).

%As shown in Fig.\ref{f:timetype}, a job of task $\tau_i$ in Fig.\ref{f:DAG} is scheduled on 3 processors. The area labeled by red line is blocking area, and the area labeled by blue line is the idle area, and the rest area during $[0,8]$ on all processors is the working area. 
%\guan{an example and figure needed here, to illustrate the the concept means the product of time and the number of processors, not only about time}

Given $m_i$ processors assigned to task $\tau_i$, we have:
\begin{lemma}\label{l:responsetime}
$\tau_i$'s worst-case response time $R_{i}$ is bounded by:
\begin{equation}
R_{i} \leq \frac{B_{i}+\Gamma_{i}+\sumC_i}{m_i}.
\label{eq:responestime}
\end{equation}	
\end{lemma}

\begin{proof}
The response time of $J_i$ is $ f(J_i) - r(J_i)$. By
$m_i \cdot (f(J_i) - r(J_i)) = B_i + W_i + \Gamma_i$ and $W_i \leq \sumC_i$, we know $J_i$'s response time is bounded by $\frac{B_{i}+\sumC_i+\Gamma_{i}}{m_i}$.
Since $J_i$ is an arbitrary job of $\tau_i$, $R_i$ is also bounded by $\frac{B_{i}+\sumC_i+\Gamma_{i}}{m_i}$.
\end{proof}
By Lemma \ref{l:responsetime}, 
the problem of bounding $R_i$  boils down to bounding $B_i + \Gamma_i$. Before going further into the analysis, we first introduce the concept of \emph{key path}:

\begin{define}[\textbf{Key Path}]
	A \emph{key path} of job $J_{i}$, denoted by $\lambda = \{v_1,v_2,...v_p\}$, is a complete path in $G_i$, s.t., $\forall j: 1 < j \leq p$, $v_{j-1}$ is a predecessor of $v_{j}$ with the \emph{latest} finish time among all predecessors of $v_{j}$.
%	\begin{itemize}
%		\item $\forall j: 1 < q \leq p$, $v_{j-1}$ is a predecessor of $v_{j}$ with the latest finish time among all predecessors of $v_{j}$.
%	\end{itemize} 
	\label{de:keypath}
\end{define}  

\begin{lemma} 
	Let $\lambda = \{v_1,v_2,...v_p\}$ be a key path of $J_i$. All $m_i$ processors must be busy at any time point in $[r(J_i), f(J_i))$ when no processor is 
busy with vertices in $\lambda$. 
	\label{l:keypath}
\end{lemma}

\begin{proof}
Let $v_{j}$ and $v_{j+1}$ be two successive elements in $\lambda$. 
	 By Definition \ref{de:keypath}, all predecessors of $v_{j+1}$ have finished at the finish time of $v_{j}$ (and thus $v_{j+1}$ is eligible for execution at that time point). 
	Therefore, all processors must be busy 
	between the finish time of $v_{j}$ and the starting time of $v_{j+1}$. Applying the above reasoning to each pair of successive elements in $\lambda$, the lemma is proved.
\end{proof}

%We classify all the time intervals 
%in $[r(J_i), f(J_i))$ into two types according to whether $\lambda$ is executing:
%\begin{itemize}
%	\item $\lambda$-busy-interval. 
%	\item $\lambda$-idle-interval. 
%\end{itemize}

In the following, we divide the blocking time $B_i$ into several disjoint parts.
There are two dimensions to divide 
$B_i$. First, we can divide 
$B_i$ into:

\begin{itemize}
	\item \textbf{Key Path Blocking Time} $B_{i}^{{\lambda}}$, the cumulative length of time spent on spinning by a vertex in $\lambda$. 
	
	\item \textbf{Delay Blocking Time} $B_{i}^{\overline{\lambda}}$, the cumulative length of time 
	on all $m_i$ processors spent on spinning during all the subintervals in  $[r(J_i), f(J_i))$ when \emph{no} processor is
	busy with a vertex in $\lambda$. 
	
	\item \textbf{Parallel Blocking Time}
	$B_{i}^{\widetilde{\lambda}}$, the cumulative length of time on all other $m_i-1$ processors spent on spinning 
	during all the subintervals in  $[r(J_i), f(J_i))$ when \emph{one} processor is
	busy with a vertex in $\lambda$.
\end{itemize}

In the second dimension we divide $B_i$
according to whether the processor is waiting for a resource locked by the \emph{same} task or by a
\emph{different} task:
 \begin{itemize}
 	\item \textbf{Intra-task Blocking Time}, the cumulative length of time spent on spinning and waiting for a resource locked by the \emph{same} task,
 	\item  \textbf{Inter-task Blocking Time},
 	 the cumulative length of time spent on
 	 spinning and waiting for a resource locked by \emph{other} tasks,
 \end{itemize} 
so each of 
$B_{i}^{{\lambda}}$, $B_{i}^{\overline{\lambda}}$ and $ B_{i}^{\widetilde{\lambda}}$ can be further divided into:
\[B_{i}^{{\lambda}} = B_{i}^{{\lambda}, I} + B_{i}^{{\lambda}, O};~ B_{i}^{\overline{\lambda}} = B_{i}^{\overline{\lambda}, I} + B_{i}^{\overline{\lambda}, O};~B_{i}^{\widetilde{\lambda}} = 
B_{i}^{\widetilde{\lambda}, I} + B_{i}^{\widetilde{\lambda}, O}\]
where the superscript $I$ denotes \emph{intra}-task blocking time and $O$ denotes \emph{inter}-task blocking time. 
Finally, $B_i$ can be divided into the following six disjoint parts:
\begin{equation}
B_i=B_{i}^{\lambda,I}+B_{i}^{\lambda,O}+B_{i}^{\overline{\lambda},I}+ B_{i}^{\overline{\lambda},O} +B_{i}^{\widetilde{\lambda},I}+B_{i}^{\widetilde{\lambda},O}.
\label{eq:totalblock}
\end{equation}

\begin{figure}[htb]
	\centering
	\includegraphics[width=2.5in]{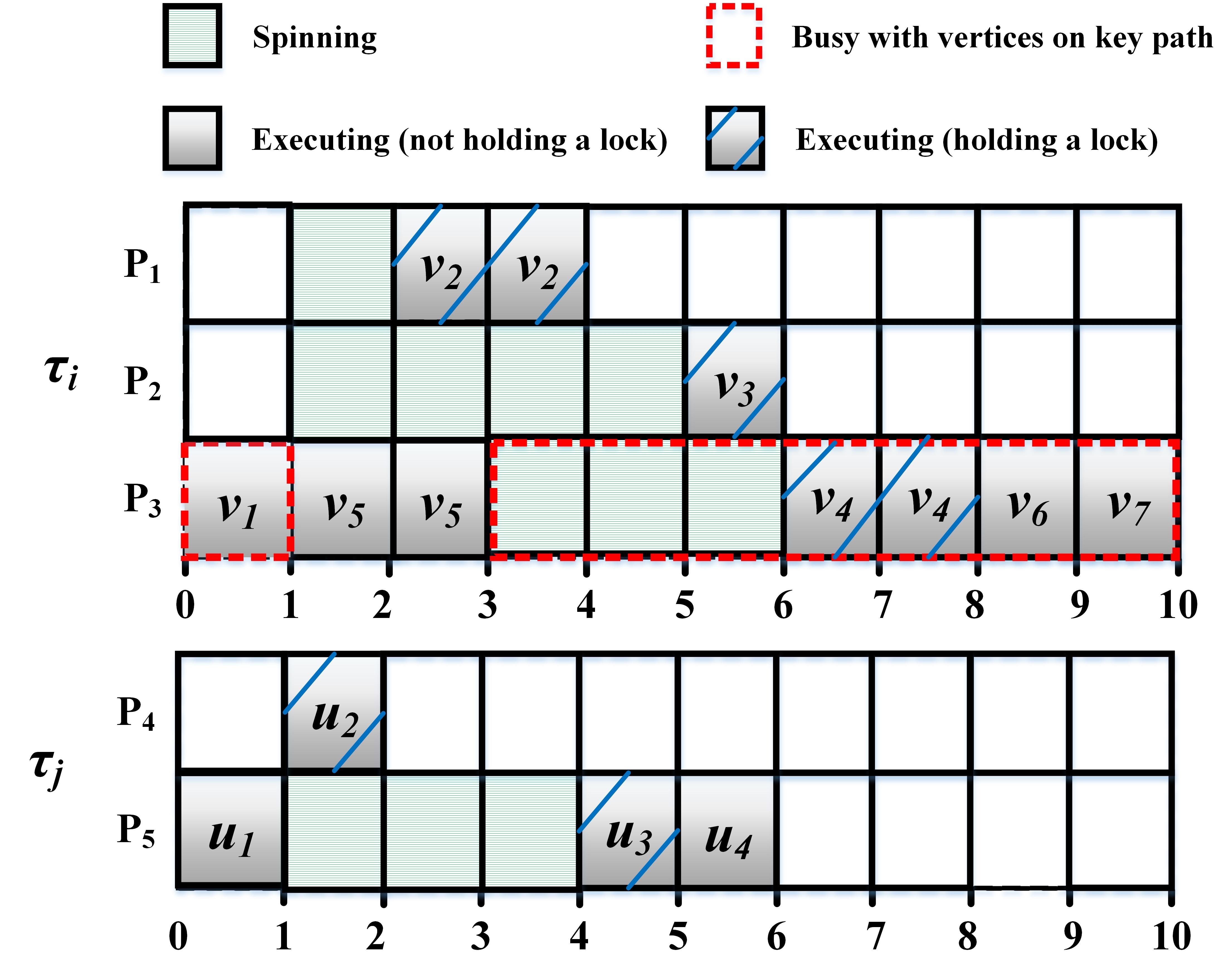}
	\caption{An example of blocking time. }
	\label{f:blocking}
\end{figure}

We use the example in Fig. \ref{f:blocking}
to demonstrate different types of blocking time. Suppose the upper part  of the figure is a running sequence of a job of the task $\tau_i$ in Fig. \ref{f:DAG}.
Suppose its key path is $\lambda =\{v_1, v_4, v_6, v_7\}$.
The lower part in Fig. \ref{f:blocking} is a running sequence of a job of another task $\tau_j$ with $V_j = \{u_1, u_2, u_3, u_4\}$ and $E_j = \{(u_1, u_2), (u_1, u_3), (u_2, u_4), (u_3, u_4)\}$.
All vertices of $\tau_j$ have the same WCET of $1$.
Entire vertices $v_2$, $v_3$ and $v_4$ in $\tau_i$ and $u_2$, $u_3$ in $\tau_j$ access the same shared resource.
The blocks wrapped by the red dash lines represent that 
a vertex in the key path is executing or spinning.
In this example, $B_i$ is divided into the six disjoint parts as follows:
	\begin{itemize}
		\item $B_{i}^{\lambda,I} = 2$, which includes $[3,4)$ and $[5,6)$ on $P_3$,
		\item $B_{i}^{\lambda,O} = 1$ which includes $[4,5)$ on $P_3$,
		\item $B_{i}^{\overline{\lambda},I} = 1$, which includes $[2,3)$ on $P_2$,
		\item $B_{i}^{\overline{\lambda},O} = 2$, which includes $[1,2)$ on both $P_1$ and $P_2$,
		\item $B_{i}^{\widetilde{\lambda},I} = 1$, which includes $[3,4)$ on $P_2$,
		\item $B_{i}^{\widetilde{\lambda},O} = 1$, which includes $[4,5)$ on $P_2$.
	\end{itemize}

%\guan{??? something is wrong with this example ???} \todo{did not find any problem of this example, what is wrong?}
%\guan{!!!unresolved!!!}

%. Let a job of $\tau_i$ and a job of $\tau_j$ are both released at time 0. One possible execution sequence is shown in Fig.\ref{f:blocking} where the area on a processor labeled by red line is the occupied by the key path. We can find time in $[1,2]$ on both both $P_1$ and $P_2$ are inter-task delay blocking. Time in $[2,3]$, $[3,4]$ and $[4,5]$ on $P_2$ are intra-task delay blocking , intra-task parallel blocking and inter-task parallel blocking, respectively. Time in $[3,4]$ and $[5,6]$ on $P_3$ are intra-task key path blocking, and $[4,5]$ on $P_3$ is inter-task key path blocking. 

%From Lemma \ref{l:delayblocking}, we know any request (or a part of request) can not introduce key path blocking and delay blocking simultaneously. Although a request may introduce key path blocking and parallel blocking at the same time, the upper bound of response time of $J_i$ given in the following theorem implies that the parallel blocking is unconcerned with the response time of $J_i$ (which can also be observed in Fig.\ref{f:blocking} where the parallel blocking $[3,5]$ on $P_2$ does not interfere the execution of the key path). 

\begin{lemma} \label{t:tightbound}
	The response time of $J_i$ is upper bounded by:
	\begin{equation}
	R(J_i)\leq \frac{\sumC_i+(m_i-1)\cdot \maxL_i +\interfer_{i}}{m_i}
	\label{eq:upperboundofr}
	\end{equation}
where	
$\interfer_{i}=(m_i-1)\cdot B_{i}^{\lambda,I}+B_{i}^{\overline{\lambda},I}+m_i \cdot  B_{i}^{\lambda,O}+B_{i}^{\overline{\lambda},O}$.
	
%	where $\sumC_i$ is the total WCET of all vertices in $G_i$, $\maxL_i$ is the longest length among all paths in $G_i$ and
%	\[
%	\interfer_{i}=(m_i-1)B_{i}^{\lambda,I}+B_{i}^{\overline{\lambda},I}+m_i B_{i}^{\lambda,O}+B_{i}^{\overline{\lambda},O}.
%	\]
\end{lemma}
\begin{proof}
We start by deriving an upper bound for $\Gamma_i$.
We use $len^*$ to denote sum of lengths of subintervals in $[r(J_i), f(J_i))$ during which a processor is busy with a vertex in $\lambda$ (i.e., a vertex in $\lambda$ is either executing or spinning). 
By Lemma \ref{l:keypath}, we know a processor can be idle only in these subintervals on $m_i - 1$ processors, so $\Gamma_i$ is bounded by $len^* \cdot (m_i - 1) $. 
Moreover, the area $len^* \cdot (m_i - 1) $ may not completely be idle time.
Some vertex may be executing/spinning in parallel with the execution/spinning of vertices in the key path $\lambda$, which can be excluded from $len^* \cdot (m_i - 1) $ 
to get a tighter upper bound for $\Gamma_i$.
In particular, we can 
subtract the following blocking time from $len^* \cdot (m_i - 1) $ to still safely bound $\Gamma_i$:
\begin {itemize}
\item  The 
parallel blocking time $B_{i}^{\widetilde{\lambda}} = B_{i}^{\widetilde{\lambda},I}+B_{i}^{\widetilde{\lambda},O}$. This type of blocking time occurs in parallel with the execution/spinning of vertices in $\lambda$, which can be excluded from the area 
$len^* \cdot (m_i - 1) $.

\item The intra-task key path blocking time $B_{i}^{{\lambda},I}$.  When some vertex in $\lambda$
is experiencing intra-task blocking, there must be a vertex in the same task $\tau_i$ holding the corresponding lock, so the same amount of time 
as $B_{i}^{{\lambda},I}$ should be excluded from the area $len^* \cdot (m_i - 1) $.
%This clearly also should be excluded 
%
%At any time point counted in $B_{i}^{{\lambda},I}$, there is at least one processor being busy (holding the resource $\resource_q$ that $v$ is spinning for). Therefore, such time should be  excluded when counting the idle time.
\end{itemize}
By the above discussion, we can get
\begin{equation}\label{e:tightbound-1}
\Gamma_i \leq len^* \cdot (m_i - 1) - B_{i}^{{\lambda},I} - B_{i}^{\widetilde{\lambda},I}-B_{i}^{\widetilde{\lambda},O}
\end{equation}
(An example illustrating the upper bound for $\Gamma_i$ is provided after the proof.)

On the other hand, we know $len^*$ is the sum of $len(\lambda)$ and the total amount of time when some vertex in $\lambda$ is spinning (for a resource held by either the same task or a different task), i.e.,
\begin{equation}\label{e:tightbound-2}
len^* = len(\lambda) + B_{i}^{{\lambda},I}+B_{i}^{{\lambda},O}
\end{equation}
%\todo{$B_{i}^{\overline{\lambda},I}+B_{i}^{\overline{\lambda},O}$ should be $B_{i}^{{\lambda},I}+B_{i}^{{\lambda},O}$}

By (\ref{e:tightbound-1}) and (\ref{e:tightbound-2}) we have:
	\begin{align*}
	& \Gamma_{i} \leq (m_i\!\!-\!\!2)\cdot B_{i}^{\lambda,I}\!\!+\!\!(m_i\!-\!1)\cdot(B_{i}^{\lambda,O}\!\!+\!len(\lambda))\!\!-\!\!(B_{i}^{\widetilde{\lambda},I}\!\!+\!\!B_{i}^{\widetilde{\lambda},O}) \\
	&	~~~~ \Rightarrow ~~~~ B_i+\Gamma_i\leq (m_i-1)\cdot len(\lambda)+\interfer_{i} ~~~~ (\text{by (\ref{eq:totalblock})})
	\end{align*}
and by Lemma \ref{l:responsetime} the lemma is proved. 
\end{proof}

Now we give the intuition of the upper bound for $\Gamma_i$ in the above proof.   
For example, as shown in Fig. \ref{f:blocking}, a job of $\tau_i$ is released at time 0 and finished at time 10. During intervals $[0,1)$ and $[3,10)$, a vertex in the key path $\lambda=\{v_1,v_4,v_6,v_7\}$ is either executing or spinning. 
A processor can be idle only in these two time intervals. 
Since $len(\lambda) = 5$, $B_{i}^{{\lambda},I} = 2$ and 
$B_{i}^{{\lambda},O} = 1$, 
so $len^*=5+ 2 + 1 = 8$, which equals the sum of the length of intervals $[0,1)$ and $[3,10)$.
Therefore, the gross upper bound for $\Gamma_i$
that counts the total area in all the time intervals on all the processors in parallel with the execution/spinning of vertices in $\lambda$ is
 $len^* \times (m_i - 1) = 16$.
In the following we show that part of this total area can be excluded to bound $\Gamma_i$.
$[3,4)$ and $[5,6)$ on $P_3$ are the intra-task key path blocking time.  $P_1$ is holding the lock in 
$[3,4)$ and $P_2$ is holding the lock in $[5,6)$, 
so we can subtract $2$ units when counting the idle time. 
$P_2$ is spinning during $[3,5)$ ($[3,4)$ is intra-task parallel blocking time and $[4,5)$ is inter-task parallel blocking time), 
so we can subtract another $2$ units when counting the idle time. 
 In summary, the idle time $\Gamma_i$ is bounded by $\Gamma_i\leq 16-2-2=12$.  
%
%
%
%\guan{the above conclusion is independent from the queuing order of spin locks ...}

%Thus in the following, we focus on the analysis of the interfering blocking which is only concerned with the key path blocking and the delay blocking. 

%\begin{theorem} \label{t:tightbound}
%	The response time:
%	\begin{equation}
%	R_{i}\leq \frac{C_i+(m_i-1)len(G_i)+B_{i}^{*}}{m_i}
%	\label{eq:upperboundofr}
%	\end{equation}
%	where
%	\[
%	B_{i}^{*}=\max_{\lambda}\{(m_i-1)B_{i}^{\lambda,I}+B_{i}^{\overline{\lambda},I}+m_i B_{i}^{\lambda,O}+B_{i}^{\overline{\lambda},O} \}.
%	\]
%\end{theorem}

%\begin{theorem}
%The upper bound of response time in Theorem \ref{t:tightbound} is tight. 	
%\end{theorem}
%\begin{proof}
%each task could be modeled into a structure where no normal workload will be executed parallel with the key path.
%\end{proof}
From Lemma \ref{t:tightbound}, the \emph{parallel blocking time} does \emph{not} contribute to the total work that may delay the finishing time of a parallel task, and the analysis is now boiled down to bounding $\interfer_i$ constituted by key path blocking time and delay blocking time. 

\section {Blocking Analysis} \label{s:firstmodel}
By adopting results we presented in Section \ref{s:pre}, in the following we develop blocking analysis techniques for three request serving orders. We define the  contribution to $\interfer_i$ by each individual resource $\resource_q$, caused by intra- and inter-task blocking, respectively: 
\begin{align}
\interfer_{i,q}^I & =(m_i-1)B_{i,q}^{\lambda,I}+B_{i,q}^{\overline{\lambda},I} \label{eq:iiqi}\\
\interfer_{i,q}^O & =m_i B_{i,q}^{\lambda,O}+B_{i,q}^{\overline{\lambda},O} \label{eq:iiqo}
\end{align}

We use $B_{i,q}^{\lambda,O,j}$ and $B_{i,q}^{\overline{\lambda},O,j}$ to denote the inter-task key path blocking time and delay blocking time on $\resource_q$ of $\tau_i$ caused by requests from task $\tau_j$ respectively where $\tau_j \neq \tau_i$, and we have:
\begin{align*}
B_{i,q}^{\lambda,O}=\sum_{j \neq i} B_{i,q}^{\lambda,O,j} \\
B_{i,q}^{\overline{\lambda},O}= \sum_{j \neq i}B_{i,q}^{\overline{\lambda},O,j}
\end{align*}
Then we divide the contribution to $\interfer_{i,q}^O$ by each individual task $\tau_j \neq \tau_i$:

\[
\interfer_{i,q}^O=\sum_{j\neq i} \interfer_{i,q}^{O,j}=\sum_{j\neq i}\left(m_i B_{i,q}^{\lambda,O,j}+B_{i,q}^{\overline{\lambda},O,j}\right).
\]

Then $\interfer_i$ can be written as 	
\begin{equation}
\interfer_{i} = \sum_{ \resource_q \in \Theta_i}  \left(\interfer_{i,q}^I + \interfer_{i,q}^O \right)= \sum_{ \resource_q \in \Theta_i}  \left(\interfer_{i,q}^I + \sum_{j \neq i}\interfer_{i,q}^{O,j} \right)
\label{eq:totalinterfering}
\end{equation}

We use $x$ to denote the number of accesses 
to resource $\resource_q$ by vertices in the key path $\lambda$. We know $x$ is in the scope $[0, N_{i,q}]$, but do not know its exact value. 
We define $\interfer_{i,q}^I(x)$ and $\interfer_{i,q}^O(x)$ as the parameterized versions of $\interfer_{i,q}^I$ and $\interfer_{i,q}^O$ with respect to $x$ respectively, then 
\[
\interfer_{i} \leq \sum_{ \resource_q \in \Theta_i} \max_{x \in [0, N_{i,q}]} \left(\interfer_{i,q}^I(x)+\interfer_{i,q}^O(x)\right)
\].

In the following, for different access polices we bound $\interfer_{i,q}^I(x)$ and $\interfer_{i,q}^O(x)$ with a particular $x$, with which we then bound $\interfer_i$. 

\subsection{Unordered}

%In this section, we first analyze the interfering blocking of $J_i$ under the first model introduced in Section \ref{??}. 
%Recall that in \textbf{Model I} we know $Q_{i,q}$, the maximal total access time to resource $\resource_q$ by all vertices of task $\tau_i$. Since we do not know the number of requests to each resource from a task, it is not easy to analyze the worts-case blocking behavior under a particular order. However, 

 We first develop analysis techniques that are applicable without distinguishing the specific order in which requests are served.   

%We  use $Q_{i,q}^*$ to denote the total access time to resource $\resource_q$ by vertices 
%\emph{not} in the key path $\lambda$.

\begin{lemma} \label{l: interferingintra}
$\interfer_{i,q}^I \leq (m_i-1) (N_{i,q} - x)  L_{i,q}$.
\end{lemma}

\begin{proof}
%	
%	\guan{first show that an access can only cause either key path blocking or ...??}
%	
%	\guan{anything to do with the FIFO order?}
The total access time to resource $\resource_q$ by vertices of $J_i$
\emph{not} in the key path $\lambda$ is at most $(N_{i,q} - x)  L_{i,q}$ which can be divided into two disjoint parts, i.e., $(N_{i,q} - x)  L_{i,q} = X + Y$, where
	\begin{itemize}
		\item $X$ is the total access time to resource $\resource_q$ that causes 
		key path blocking. We know \begin{equation}\label{e: interferingintra-1}
	B_{i,q}^{\lambda, I} = X 
		\end{equation}
		
		\item $Y$ is the total access time to  
$\resource_q$ that does \emph{not} cause key path blocking.
By definition, key path blocking 
and delay blocking cannot happen at the same time.
Therefore, any lock holding time that causes 
intra-task delay blocking must be included in $Y$.
Each time unit in $Y$ can cause at most $(m_i - 1)$
intra-task delay blocking time (one processor is holding the lock and at most $m_i - 1$ processors are spinning). In summary, the intra-task delay blocking is bounded by \begin{equation}\label{e: interferingintra-2}
B_{i,q}^{\overline{\lambda}, I} \leq (m_i - 1) Y
\end{equation}
\end{itemize}
By (\ref{e: interferingintra-1}) and (\ref{e: interferingintra-2}) we have

%Since $N_{i,q} - x$ is the number of accesses 
%to $\resource_q$ by vertices not in  $\lambda$ and the maximal lock holding time of each of these accesses is bounded by $L_{i,q}$, we know 
%$Q_{i,q} \leq (N_{i,q} - x)  L_{i,q} $.

\begin{align*}
	\interfer_{i,q}^I  &= (m_i -  1) B_{i,q}^{\lambda, I} + B_{i,q}^{\overline{\lambda}, I} \leq (m_i - 1)(X + Y)
	B_{i,q}^{\overline{\lambda},I} \\ 
	&=(m_i \!-\! 1)(N_{i,q} - x)  L_{i,q} 
\end{align*}

The lemma is proved.
\end{proof}

Lemma \ref{l: interferingintra} directly implies:
\begin{corollary}\label{c:inter}
$\interfer_{i,q}^I \leq (m_i-1)N_{i,q}L_{i,q}$
\end{corollary}	
%Actually, in  \textbf{Model I}
%we are not able to make any difference between 
%$Q_{i,q}^*$ and $Q_{i,q}$, and later we will  use Corollary \ref{c:inter} to derive the schedulability  analysis condition for  \textbf{Model I}. 
%Nevertheless, we still keep the original result Lemma \ref{l: interferingintra} as it will be reused in the analysis for \textbf{Model II} in the next section.
In the following we bound $\interfer_{i,q}^O$.  We use $\njobs_{i,j}^q$ to denote the maximal number of jobs of $\tau_j$
that may have contention on  resource $\resource_q$ with the analyzed job $J_i$ of task $\tau_i$, which can be computed 
by \cite{yang2015global,brandenburg2011scheduling}:
\begin{equation}\label{eq:nijq}
\njobs_{i,j}^q=
\left\{
\begin{array}{ll}
 \lceil\frac{D_i+D_j}{T_j}  \rceil & \text{if both $\tau_i$ and $\tau_j$ access $\resource_q$} \\
0 & \text{otherwise}
\end{array}\right.
\end{equation}

$\njobs_{i,j}^q = 0$ if either $\tau_i$ or $\tau_j$ does not access $\resource_q$,
since there is no inter-task blocking between $\tau_i$ and $\tau_j$ due to 
$\resource_q$.

\begin{lemma} \label{l: interferinginter}
$\displaystyle 
\interfer_{i,q}^{O,j}\leq  m_i \njobs_{i,j}^q N_{j,q} L_{j,q}
$
\end{lemma}

\begin{proof}
%If either $\tau_i$ or $\tau_j$ does not access resource $\resource_q$, there is no inter-task blocking between $\tau_i$ and $\tau_j$ due to 
% $\resource_q$, so both $B_{i,q}^{\lambda,O}$ and  $B_{i,q}^{\overline{\lambda},O}$ are zero, so $	\interfer_{i,q}^{O} = 0$. In the following we focus on the case that both $\tau_i$ and $\tau_j$ access resource $\resource_q$.
%    	$\interfer_{i,q}^{O}$ consists of the inter-task interfering blocking all other tasks in $\tau$ except $\tau_i$. Let $\interfer_{i,q}^{O,j}$ denote the inter-task interfering blocking from $\tau_j$. Then we have $\interfer_{i,q}^{O}=\sum_{j\neq i}\interfer_{i,q}^{O,j}$. We first derive the upper bound of $\interfer_{i,q}^{O,j}$. 

The maximum number of jobs of $\tau_j$ that may contend with $J_i$ on $\resource_q$ is $\njobs_{i,j}^q$.
The total access time to $\resource_q$ by all other jobs of $\tau_j$ during $[r(J_i), f(J_i))$ 
is at most
$\njobs_{i,j}^q N_{j,q} L_{j,q}$. We divide it into two disjoint parts $\njobs_{i,j}^q N_{j,q} L_{j,q} = X + Y$, where:
	\begin{itemize}
		\item $X$ is the total access time  to resource $\resource_q$ by $\tau_j$ that causes 
		key path blocking. We know \begin{equation}\label{e: interferinginter-1}
		B_{i,q}^{\lambda, O,j} = X 
		\end{equation}
		
		\item $Y$ is the total access time  to  
		$\resource_q$ by $\tau_j$ that does \emph{not} cause key path blocking. By definition, key path blocking 
		and delay blocking cannot happen at the same time.
		Therefore, any resource access time that causes 
		inter-task delay blocking must be included in $Y$.
		Each time unit in $Y$ can cause at most $m_i$
		inter-task delay blocking time (at most $m_i$ processors are spinning). Therefore, the inter-task delay blocking is bounded by
		
		\begin{equation}\label{e: interferinginter-2}
		B_{i,q}^{\overline{\lambda}, O.j} \leq m_i Y
		\end{equation}
	\end{itemize}
By (\ref{e: interferinginter-1}) and (\ref{e: interferinginter-2}) we have
\begin{equation*}
\interfer_{i,q}^{O,j} = m_i B_{i,q}^{\lambda, O,j} + B_{i,q}^{\overline{\lambda}, O,j} \leq m_i (X + Y) = m_i \njobs_{i,j}^q N_{j,q} L_{j,q}\vspace{-6mm}
\end{equation*}
%The maximum number of jobs of $\tau_j$ that may contend with $J_i$ on $\resource_q$ is $\njobs_{i,j}^q$. Then by summing up the accessing time to $\resource_q$ of all tasks that may contend with $J_i$ on $\resource_q$, we have $X+Y \leq \sum_{j\neq i} \njobs_{i,j}^q Q_{j,q}$. Then by (\ref{e: interferinginter}), the lemma is proved.
\end{proof}

Now we are ready to bound $\tau_i$'s worst-case response time.

\begin{theorem} \label{t: test1}
	For unordered, $R_i$ is bounded by: 
	%\begin{equation}
	%\wcrt_i \leq \frac{C_i+(m_i-1)(\critical_i+\sum_{\resource_{q}\in \resourceset_i}\totolresource_{i,q})+m_i\sum_{j \neq i}\njobs_{i,j} \sum_{\resource_q\in \resourceset_i}\totolresource_{j,q}}{m_i}
	%\end{equation}
	\begin{equation*}
	R_i \leq 
	\frac{\sumC_i\!+\!(m_i \!-\! 1)(\critical_i\!+\!\displaystyle\sum_{\resource_{q}\in \resourceset_i}N_{i,q} L_{i,q})}{m_i}	
	+
	\sum_{j \neq i}\sum_{\resource_q \in \resourceset_i} \!\! \njobs_{i,j}^q N_{j,q} L_{j,q}
	\end{equation*}
%	where $\totolresource_{i,q}=N_{i,q} L_{i,q}$ and $\totolresource_{j,q}=N_{j,q} L_{j,q}$.
\end{theorem}

\begin{proof}
	
	By condition (\ref{eq:totalinterfering}), Corollary \ref{c:inter} and Lemma \ref{l: interferinginter}, we have 
%	\begin{equation}
%	\interfer_{i,q}\leq  (m_i-1)Q_{i,q}+ \sum_{j\neq i}m_i \njobs_{i,j}^q \totolresource_{j,q}
%	\label{eq:inequality}
%	\end{equation}	
%	By summing up this for all resources accessed by $J_i$, we have
		\begin{equation}
		\interfer_{i}\leq \sum_{\resource_q \in \resourceset_i}
		\left( (m_i-1)N_{i,q} L_{i,q}+m_i \sum_{j\neq i}\njobs_{i,j}^qN_{j,q} L_{j,q} \right)
		\label{eq:overallinterfering}
		\end{equation}
and by Lemma \ref{t:tightbound} the theorem is proved.
\end{proof}
 
Task $\tau_i$ is schedulable if $R_i \leq D_i$, so we can  calculate the value of $m_i$ for $\tau_i$ to be schedulable based on Theorem \ref{t: test1}: 
\begin{corollary}\label{c:mallocation}
	Task $\tau_i$ is schedulable on $m_i$ processors if 
	\begin{equation}\label{eq:co21}
	D_i-(\sum_{j \neq i}\njobs_{i,j}^q\sum_{\resource_q \in \resourceset_i}N_{j,q} L_{j,q}\!+\!\critical_i\!+\!\sum_{\resource_{q}\in \resourceset_i}N_{i,q} L_{i,q})>0
	\end{equation}
	and
	\[
	m_i=\left \lceil \frac{\sumC_i-(\critical_i+\displaystyle\sum_{\resource_{q}\in \resourceset_i}N_{i,q} L_{i,q})}{D_i-(\displaystyle\sum_{j \neq i}\njobs_{i,j}^q\displaystyle\sum_{\resource_q \in \resourceset_i}N_{j,q} L_{j,q}\!+\!\critical_i\!+\!\displaystyle\sum_{\resource_{q}\in \resourceset_i}N_{i,q} L_{i,q})} \right \rceil
	\]	
\end{corollary}

If each task can get enough processors according to Corollary \ref{c:mallocation}, the whole system is schedulable. Otherwise, the system is decided to be unschedulable. This procedure is shown in Algorithm \ref{al:firstmodel}.

%
%Both $\sumC_i$ and $\maxL_i$ 
%can be computed in linear time with respect to $|V_i|$, the number of vertices in $G_i$ \cite{cormen2001introduction}, 
%so the complexity to compute $m_i$ for each $\tau_i$ 
%is $O(|V_i|+|\Theta_i|\cdot|\Tau|)$,
%where $|\Theta_i|$ is the set of resources accessed by $\tau_i$ and $|\Tau|$ is the number of tasks in the system.
%Therefore, the overall complexity of Algorithm \ref{al:firstmodel} is $O(\mathcal{V}_{max}\cdot|\Tau|+ |\Theta|\cdot|\Tau|^2)$,
%where $\mathcal{V}_{max}$
%is the maximal number of vertices of a task, $\Theta$ and $\Tau$ are the set of resources and set of tasks in the system, respectively.

\begin{algorithm}
	\caption{Processor partitioning algorithm for unordered.}
	\begin{algorithmic}[1]
		\FOR{each task $\tau_i \in \Tau$}
		\IF{(\ref{eq:co21}) is satisfied}
		\STATE calculate $m_i$ according to Corollary \ref{c:mallocation};
		%		\[
		%		m_i=\lceil \frac{C_i-(\critical_i+\sum_{\resource_{q}\in \resourceset_i}\totolresource_{i,q})}{D_i-(\sum_{j \neq i}\njobs_{i,j}\sum_{\resource_q \in \resourceset_i}\totolresource_{j,q}+\critical_i+\sum_{\resource_{q}\in \resourceset_i}\totolresource_{i,q})}\rceil
		%		\]
		\IF{less than $m_i$ processors are available}
		\STATE \textbf{return} \textit{unschedulable}
		\ENDIF
		\STATE assign $m_i$  processors to $\tau_i$
		\ELSE
		\STATE \textbf{return} \textit{unschedulable}
		\ENDIF
		\ENDFOR
		\STATE \textbf{return} \textit{schedulable} 
	\end{algorithmic}
	\label{al:firstmodel}
\end{algorithm}

\subsection{FIFO-order}
In the following we develop analysis techniques for FIFO-order. We first derive an upper bound for $\interfer_{i,q}^I(x)$ with a particular $x$:

\begin{lemma} \label{l: improveinterferingintra}
%	The intra-task interfering blocking $\interfer_{i,q}^I(x)$ is bounded by: 
$  \interfer_{i,q}^I(x) \leq \mathcal{F}^I(x)$ in FIFO-order, where
	\[
\mathcal{F}^I(x) = ((N_{i,q}-x^{})(m_i-1)-\max\{1-x, 0\}\Delta)L_{i,q}
	\]
	and $\Delta=\min\{N_{i,q}, m_i\}\left (m_i-\frac{\min\{N_{i,q}, m_i\}+1}{2} \right)$.
\end{lemma}
\begin{proof}
%Let $\interferset=\{\phi_1,\phi_2,\cdots,\phi_{|\interferset|}\}$ denote the set of accessing time for $\resource_q$ from vertices not in $\lambda$. 
We prove the lemma in two cases.

\begin{enumerate}
\item $x \neq 0$. By Lemma \ref{l: interferingintra} we know
for any $x$ it holds:
%Since $N_{i,q} - x$ is the number of accesses 
%to $\resource_q$ by vertices not in  $\lambda$ and the maximal lock holding time of each of these accesses is bounded by $L_{i,q}$, we know 
%$Q_{i,q}^* \leq (N_{i,q} - x)  L_{i,q} $. This, together with 
%(\ref{e: improveinterferingintra-1}), gives
%\begin{equation}
%\interfer_{i,q}^I(x) \leq (N_{i,q}-x)(m_i-1)L_{i,q}
%\label{eq:notzero}
%\end{equation}
%which proves the lemma for $x\neq 0$.

%\todo{to be revised by adopting FIFO properties}
%With FIFO spin locks, each resource access request of $J_i$ for $\resource_q$ is blocked by at most $m_i-1$ requests from $\tau_i$. We know 
%
%\begin{equation}\label{e: FIFO-KEY}
%B_{i,q}^{\lambda, I} = x (m_i-1)L_{i,q} 
%\end{equation}
%
%and 
%
%\begin{equation}\label{e: FIFO-NONKEY}
%B_{i,q}^{\overline{\lambda}, I} \leq (m_i - 1) (N_{i,q}-x)L_{i,q}
%\end{equation}
%
%
%By (\ref{e: FIFO-KEY}) and (\ref{e: FIFO-NONKEY}) we have
%\[
%\interfer_{i,q}^I  = (m_i \!- \! 1) B_{i,q}^{\lambda, I} + B_{i,q}^{\overline{\lambda}, I} \leq (m_i \!-\! 1)(X \!+\! Y) = (m_i \!-\! 1)(N_{i,q} - x)  L_{i,q}
%\]
%
%

\begin{equation}\label{e: improveinterferingintra-1}
\interfer_{i,q}^I(x) \leq (N_{i,q}-x)(m_i-1)L_{i,q}
\end{equation}

\item $x = 0$. In this case, 
$B_{i,q}^{\lambda,I}=0$ and  $B_{i,q}^{\overline{\lambda},I}$ is bounded by the maximum blocking time that may be introduced by $N_{i,q}$ requests on $m_i$ processors which equals $(\frac{\alpha(\alpha-1)}{2}+(m_i-1)(N_{i,q}-\alpha))L_{i,q}$ \cite{dinh2018blocking}, where $\alpha=\min\{N_{i,q}, m_i\}$.

Then we have
\begin{eqnarray*}
B_{i,q}^{\overline{\lambda},I}&\leq&(\frac{\alpha(\alpha-1)}{2}+(m_i-1)(N_{i,q}-\alpha))L_{i,q} \nonumber \\
&=&((m_i-1)N_{i,q}-\alpha(m_i-\frac{\alpha+1}{2}))L_{i,q} \\
&=&((m_i-1)N_{i,q}-\Delta )L_{i,q} 
\label{eq:allnotonkey}
\end{eqnarray*}

Therefore, when $x=0$ (thus $B_{i,q}^{\lambda,I}=0$) we have 
\begin{equation*}
\interfer_{i,q}^I(x)= 0 + B_{i,q}^{\overline{\lambda},I}\leq ((m_i-1)N_{i,q}-\Delta )L_{i,q}
\label{eq:zero}
\end{equation*}
\end{enumerate}
In summary, in both cases the lemma is proved.
\end{proof}

\begin{lemma} \label{l: improveinterferinginter}
$\interfer_{i,q}^O(x) \leq \mathcal{F}^O(x)$ in FIFO-order, where
		\[ 
	\mathcal{F}^O(x) = \sum_{j\neq i}	\min\{ m_i \njobs_{i,j}^q N_{j,q}, (N_{i,q}+(m_i-1)x)m_j\}L_{j,q}		
		\]
\end{lemma}

\begin{proof}
From Lemma \ref{l: interferinginter}, we have:
\begin{equation}
\interfer_{i,q}^{O,j}(x) \leq  m_i \njobs_{i,j}^q N_{j,q} L_{j,q}
\label{eq:interferinter1}
\end{equation}
% Let $\interfersetj$ denote the set of requests for $\resource_q$ of jobs from $\tau_j$ which block $J_i$.
%Let $B_{i,q}^{\lambda,O,j}$ and $B_{i,q}^{\overline{\lambda},O,j}$ denote the inter-task key path blocking and the inter-task delay blocking of $J_i$ introduced by $\tau_j$, respectively. 
With FIFO spin locks, at most $m_j$ requests from  $\tau_j$ can be spinning at the same time (in the queue waiting for $\resource_q$), each request of $J_i$ for $\resource_q$ is blocked by at most $m_j$ requests from another task $\tau_j$
(at most $m_j$ requests from  $\tau_j$ are in the queue waiting for $\resource_q$), so
$B_{i,q}^{\lambda,O,j} $ for $x$  accesses to $\resource_q$ of vertices in $\lambda$ is bounded by 
\begin{align*}
B_{i,q}^{\lambda,O,j} & \leq  x m_j L_{j,q}
\end{align*}
The remaining $N_{i,q} - x$ accesses to $\resource_q$ are from vertices \emph{not} in $\lambda$, for which $B_{i,q}^{\overline{\lambda},O,j}$ is bounded by
\begin{align*}
B_{i,q}^{\overline{\lambda},O,j} & \leq  (N_{i,q}-x)m_j L_{j,q}
\end{align*}
Applying them to $\interfer_{i,q}^{O,j}(x)  =  m_iB_{i,q}^{\lambda,O,j}+B_{i,q}^{\overline{\lambda},O,j}$  gives
\begin{align*}
\interfer_{i,q}^{O,j}(x) 
& \leq  ( x m_i m_j +(N_{i,q}-x)m_j) L_{j,q} 	\\
& =  (N_{i,q}+(m_i-1)x)m_jL_{j,q}	
\end{align*}

By getting the minimum of this bound and the bound in (\ref{eq:interferinter1}), the lemma is proved.
\end{proof}

By now we have bounded both $\interfer_{i,q}^I(x)$ and $\interfer_{i,q}^O(x)$ for resource $\resource_q$ with a particular $x$. Since $x$ is unknown, we need to  find the value of $x$ in $[0, N_{i,q}]$ that leads to the maximal
$\interfer_{i,q}^I(x) + \interfer_{i,q}^O(x)$. By doing this for each $\resource_q \in \Theta_i$, we obtain an upper bound for $\interfer_i$ as follows:
\begin{lemma} \label{l: improveoverallinterfering}
	In FIFO-order, we have:
	\[
	\interfer_{i}\leq \sum_{\resource_q \in \resourceset_i}\max_{x\in [0, N_{i,q}]}\{  \mathcal{F}^I(x) +  \mathcal{F}^O(x)\}
	\]
	
\end{lemma}

%\begin{proof}
%	The prove is done by combining Lemma \ref{l: improveinterferingintra} and \ref{l: improveinterferinginter}, and summing up the maximum interfering blocking of all resources accessed by $J_i$.
%\end{proof}

Then by applying this to Lemma \ref{t:tightbound}, we can bound the worst-case response time of $\tau_i$: 

\begin{theorem} \label{l: bound2}
	In FIFO-order, $R_i$ is bounded by: 
	\[
	R_i \leq 
	\frac{\displaystyle \sumC_i+(m_i \!-\! 1)\critical_i+ \sum_{\resource_q \in \resourceset_i}\max_{x\in [0, N_{i,q}]}\{  \mathcal{F}^I(x) + \mathcal{F}^O(x)\}}{m_i}	
	\]	
	where $\mathcal{F}^I(x)$ and $\mathcal{F}^O(x)$ are defined in 
	Lemma \ref{l: improveinterferingintra}  and \ref{l: improveinterferinginter}.
\end{theorem}

\begin{proof}
	Proved by Lemma \ref{t:tightbound}  and Lemma \ref{l: improveoverallinterfering}.
\end{proof}

If the number of processors $m_i$ assigned to each task is given, we can use Theorem \ref{l: bound2} to compute $R_i$ and compare it with $D_i$ to decide the schedulability of $\tau_i$.

However, if the number of processors $m_i$ assigned to each task is not given and we are required to partition the total $m$ processors to each task, 
we are \emph{not} able to directly compute $m_i$ for each task $\tau_i$. This is because the worst-case response time bound of a task in Theorem \ref{l: bound2} (more specifically, $\mathcal{F}^{O}(x)$) depends on the number of processors assigned to \emph{other} tasks. Therefore, there is a cyclic dependency among the number of processor assigned to different tasks: to decide $m_i$ for $\tau_i$, we need to know $m_j$ for $\tau_j$, while 
to decide $m_j$ for $\tau_j$, we need to know $m_i$ for $\tau_i$.

In the following we present an algorithm to iteratively compute $m_i$ for each task $\tau_i$ in the presence of the cyclic dependency mentioned above. Initially, we set 
$m_i=\lceil \frac{\sumC_i-\critical_i}{D_i-\critical_i} \rceil$ for each $\tau_i$, which is number of processors to make $\tau_i$ schedulable without considering the shared resources \cite{li2014analysis}. This is a lower bound of our desired $m_i$. Then starting with these initial $m_i$ values, we gradually increase $m_i$ for each $\tau_i$, until finding a set of $m_i$ values for all tasks to make them all schedulable according to Theorem \ref{l: bound2}.
The pseudo-code of this procedure is presented in Algorithm 
\ref{al:secondmodel}.
%
%!!!!!!!!!!!!!!!!!!!!!!!!
%Algorithm 	\ref{al:secondmodel} 
%
%
%
%The same as in Section \ref{s:firstmodel}, we compare the upper bound of response time of $J_i$ with its relative deadline. However, we can not compute $m_i$ for each task $\tau_i$ to be schedulable directly, since the interfering blocking $\interfer_i$ of $J_i$ and $m_i$ depend on each other. Instead, we use an iteration algorithm to decide the number of processors that need to be assigned to each task (test the schedulability of task set $\tau$) which works as follows:
%
%\begin{itemize}
%	\item Step 1: For each task $\tau_i$, initializes $m_i=\lceil \frac{C_i-\critical_i}{D_i-\critical_i} \rceil$
%	\item Step 2: For each task, for each $\resource_q$, enumerates each $\nkeyresource_{i,q}^{\lambda} \in [1, N_{i,q}]$ to find $\max_{x \in [1, N_{i,q}]}\{  \interfer_{i,q}^I(x) +\interfer_{i,q}^O(x)\}$.
%	\item Step 3: Tests if $R_i \leq D_i$ satisfies for each task according to Theorem \ref{t:tightbound}. If yes, the algorithm returns success. Otherwise, for each task $\tau_i$ does not satisfy $R_i \leq D_i$, the algorithm does $m_i=m_i+1$. If there are not enough processors remaining, the algorithm returns failure, else it goes to Step 2.
%\end{itemize}
%
%The above algorithm is shown in Algorithm \ref{al:secondmodel}.
\reduceblackhalf
\begin{algorithm}
	\caption{Processor partitioning algorithm for FIFO-order.}
	\begin{algorithmic}[1]	
		%\FOR{each task $\tau_i$}
		\STATE For each $\tau_i$: $m_i \gets \lceil \frac{\sumC_i-\critical_i}{D_i-\critical_i} \rceil$;
		%\ENDFOR	
		%\WHILE{1}
		\WHILE{(1)}
		\STATE $update\gets 0$;
		\FOR{each task $\tau_i$}
		\FOR{each resource $\resource_q \in \resourceset_i$} 
		%		  \IF {FIFO order is adopted}
		\STATE \label{line:maximum} find $x\in [0, N_{i,q}]$  s.t., $\mathcal{F}^I(x) +\mathcal{F}^O(x)$ is maximal;
		%		  \ELSE
		%		  \IF {Priority Order is adopted} 
		%		  	\STATE \label{line:maximumP} find $x\in [0, N_{i,q}]$  s.t., $\mathcal{P}^I(x) +\mathcal{P}^O(x)$ is maximal;
		%		  \ENDIF
		%		  \ENDIF  
		\ENDFOR
		\STATE Compute the WCRT bound $R_i'$ using Theorem \ref{l: bound2};
		\IF{ $R_i' > D_i$}
		\STATE $m_i \gets m_i+1$; $update \gets 1$;
		\ENDIF
		\ENDFOR
		\IF{$\sum_{\tau_i\in \tau}m_i>m$}
		\STATE \textbf{return} \textit{unschedulable}
		\ENDIF
		\IF{$update=0$} 
		\STATE \textbf{return} \textit{schedulable} 
		\ENDIF
		\ENDWHILE
	\end{algorithmic}
	\label{al:secondmodel}
\end{algorithm}
%The while-loop in Algorithm	\ref{al:secondmodel} iterates for at most $m$ times, since in each iteration at least one task increases its $m_i$ by $1$ and the algorithm terminates with \textit{unschedulable} when $\sum_{\tau_i \in \tau}{m_i}$ exceeds $m$. 
%For each $x$, the time complexity to compute $\mathcal{F}^I(x) +\mathcal{F}^O(x)$ is $|\Tau|$, 
%so the time complexity to find the desired $x$ in Line \ref{line:maximum} is $O(N_{max} \cdot |\Tau|)$, where $N_{max}$
%is the maximal number of access requests to  $\resource_q$ by $\tau_i$, among all combinations of $\resource_q$ and $\tau_i$.
%The overall time complexity of Algorithm	\ref{al:secondmodel} is $O(m \cdot |\Tau|^2 \cdot |\Theta| \cdot N_{max})$.

\todo{It is necessary to mention that Algorithm \ref{al:secondmodel} is a heuristic algorithm to compute $m_i$ for each task to be schedulable. However, Algorithm \ref{al:secondmodel} is not optimal in the sense that minimum number of processors required by the whole task set to be schedulable is obtained. For example, after an iteration, there are two tasks that are not schedulable. Then according to Algorithm \ref{al:secondmodel}, the number of processors required by these two tasks are both increased by 1. However, it is possible that after the number of processors required by one of these two tasks is increased by 1, the other task becomes schedulable. Finding the optimal processor allocation algorithm is out of the scope of this paper which will be investigated in our future work.}

\subsection{Priority-Order}
In the following we develop analysis techniques for priority-order. We use $\tau_i^{H}$ and $\tau_i^{L}$ to denote the set of tasks with higher and lower priorities than $\tau_i$, respectively. 

We first bound $\interfer_{i,q}^I(x)$. Since different requests to a resource from the same task have the same priority, the upper bound of intra-task blocking time in priority-order is the same as in FIFO-order. Then we have: 

\begin{lemma} \label{l: priorityintra}
	%	The intra-task interfering blocking $\interfer_{i,q}^I(x)$ is bounded by: 
	$  \interfer_{i,q}^I(x) \leq \mathcal{P}^I(x)$ in priority-order, where
	\[
	\mathcal{P}^I(x) = ((N_{i,q}-x^{})(m_i-1)-\max\{1-x, 0\}\Delta)L_{i,q}
	\]
	and $\Delta=\min\{N_{i,q}, m_i\}\left (m_i-\frac{\min\{N_{i,q}, m_i\}+1}{2} \right)$.
\end{lemma}

\begin{proof}
The lemma is the same as the proof of Lemma \ref{l: improveinterferingintra}. 
\end{proof}

%which consists of delay by requests from tasks with different priorities: tasks with lower priorities, equally priority and higher priorities.
In the following we bound $\interfer_{i,q}^O(x)$ in priority-order. We use $\varDelta_{i,j}^q$ to denote the maximal number of jobs of $\tau_j$
that may have contention on resource $\resource_q$ with a single request from job $J_i$ of task $\tau_i$, which can be computed 
by \cite{yang2015global,brandenburg2011scheduling}:
\begin{equation}
\varDelta_{i,j}^q=
\left\{
\begin{array}{ll}
\lceil\frac{dpr(\tau_i,\resource_q)+D_j}{T_j}  \rceil & \text{if both $\tau_i$ and $\tau_j$ access $\resource_q$} \\
0 & \text{otherwise}
\end{array}\right.
\label{eq: maximumdelaynumber}
\end{equation}

where $dpr(\tau_i,\resource_q)$ is \emph{delay-per-request} \cite{dinh2018blocking} on $\resource_q$ of $\tau_i$. $dpr(\tau_i,\resource_q)$ denotes the length of time interval between the time that a request of $\resource_q$ from $\tau_i$ issues and the time it is served, which can be calculated by a fix-point iteration method (the calculation of $dpr(\tau_i, l_q)$ is the same as in \cite{dinh2018blocking}, thus omitted here).   

%To bound $\interfer_{i,q}^O(x)$ under priority order, we first refer to a concept of \emph{delay-per-request} \cite{dinh2018blocking},  which is the length of time interval between the time that a request issues and the time it is served. We denote the delay-per-request on $\resource_q$ of $\tau_i$ as $\varDelta_{i,q}$, which can be calculated by a fix-point iteration method (the calculation of $\varDelta_{i,q}$ is the same as of $dpr(\tau_i, l_q)$ in \cite{dinh2018blocking}, thus omitted here). 

\begin{lemma} \label{l: priorityinter}
	$\interfer_{i,q}^O(x) \leq \mathcal{P}_L^O(x)+\mathcal{P}_H^O(x)$ in priority-order, where
	\[
	 \mathcal{P}_L^O(x)=\left(N_{i,q}+(m_i-1)x\right)\max_{\tau_j\in \tau_i^L}\{L_{j,q}\},
	\]
	and
	\[ 
	\mathcal{P}_H^O(x)\! = \!\!\!\sum_{\tau_j\in \tau_i^H}\!\!\min\{ m_i \njobs_{i,j}^q N_{j,q}, (N_{i,q}+(m_i-1)x)\varDelta_{i,j}^q N_{j,q}\}L_{j,q}.		
	\]
\end{lemma}

\begin{proof}
	We divide $\interfer_{i,q}^O(x)$ by each individual task according to its priority: 
	\[
	\interfer_{i,q}^O(x)=\sum_{\tau_j\in \tau_i^L}\interfer_{i,q}^{O,j}(x)+\sum_{\tau_j\in \tau_i^H}\interfer_{i,q}^{O,j}(x)
	\]
		With priority ordered spin locks, each resource access request of $J_i$ for $\resource_q$ is blocked by at most one request from all tasks with lower priorities than $\tau_i$, so $\forall \tau_j \in \tau_i^L$,
	$\sum_{\tau_j\in \tau_i^L}B_{i,q}^{\lambda,O,j} $ for $x$  accesses to $\resource_q$ of vertices in $\lambda$ is bounded by 
	\begin{align*}
	\sum_{\tau_j\in \tau_i^L}B_{i,q}^{\lambda,O,j} & \leq  x \max_{\tau_j\in \tau_i^L}\{L_{j,q}\}
	\end{align*}
	The remaining $N_{i,q} - x$ accesses to $\resource_q$ are from vertices \emph{not} in $\lambda$, for which $\sum_{\tau_j\in \tau_i^L}B_{i,q}^{\overline{\lambda},O,j}$ is bounded by
	\begin{align*}
	\sum_{\tau_j\in \tau_i^L}B_{i,q}^{\overline{\lambda},O,j} & \leq  (N_{i,q}-x)\max_{\tau_j\in \tau_i^L}\{L_{j,q}\}
	\end{align*}
	Applying them to $\interfer_{i,q}^{O,j}(x)  =  m_iB_{i,q}^{\lambda,O,j}+B_{i,q}^{\overline{\lambda},O,j}$  gives
	\begin{equation}
	\sum_{\tau_j\in \tau_i^L}\interfer_{i,q}^{O,j}(x) 
	 \leq  \mathcal{P}_L^O(x).
	 \label{eq:lower}
	\end{equation}
	
	In the following, we focus on bounding $\sum_{\tau_j\in \tau_i^H}\interfer_{i,q}^{O,j}(x)$.
	
	From Lemma \ref{l: interferinginter}, we have:
	\begin{equation}
	\interfer_{i,q}^{O,j}(x) \leq  m_i \njobs_{i,j}^q N_{j,q} L_{j,q}
	\label{eq:interferinterhigh}
	\end{equation}

	From (\ref{eq: maximumdelaynumber}), each resource access request of $J_i$ for $\resource_q$ is blocked by at most $\varDelta_{i,j}^q N_{j,q}$ requests from $\tau_j$ in priority-order, so $\forall \tau_j \in \tau_i^H$,
	$B_{i,q}^{\lambda,O,j} $ for $x$  accesses to $\resource_q$ of vertices in $\lambda$ is bounded by 
	\begin{align*}
	B_{i,q}^{\lambda,O,j} & \leq  x \varDelta_{i,j}^q N_{j,q} L_{j,q}
	\end{align*}
	The remaining $N_{i,q} - x$ accesses to $\resource_q$ are from vertices \emph{not} in $\lambda$, for which $B_{i,q}^{\overline{\lambda},O,j}$ is bounded by
	\begin{align*}
	B_{i,q}^{\overline{\lambda},O,j} & \leq  (N_{i,q}-x)\varDelta_{i,j}^q N_{j,q} L_{j,q}
	\end{align*}
	Applying them to $\interfer_{i,q}^{O,j}(x)  =  m_iB_{i,q}^{\lambda,O,j}+B_{i,q}^{\overline{\lambda},O,j}$  gives
	\begin{align*}
	\interfer_{i,q}^{O,j}(x) 
	& \leq  ( x m_i \varDelta_{i,j}^q N_{j,q} +(N_{i,q}-x)\varDelta_{i,j}^q N_{j,q}) L_{j,q} 	\\
	& =  (N_{i,q}+(m_i-1)x)\varDelta_{i,j}^q N_{j,q}L_{j,q}	
	\end{align*}
	
	Getting the minimum of this bound and the bound in (\ref{eq:interferinterhigh}) gives us: 
	\[
	\sum_{\tau_j\in \tau_i^H}\interfer_{i,q}^{O,j}(x) \leq \mathcal{P}_H^O(x).
	\]
Combining with (\ref{eq:lower}), the lemma is proved.
\end{proof}

Then we can bound the worst-case response time of $\tau_i$ in priority-order: 

\begin{theorem} \label{l: bound3}
	In priority-order, $R_i$ is bounded by: 
	\[
	R_i \leq 
	\frac{\displaystyle \sumC_i+(m_i \!-\! 1)\critical_i+ \sum_{\resource_q \in \resourceset_i}\max_{x\in [0, N_{i,q}]}\{  \mathcal{P}^I(x) + \mathcal{P}^O(x)\}}{m_i}	
	\]	
	where $\mathcal{P}^I(x)$ and $\mathcal{P}^O(x)=\mathcal{P}_L^O(x)+\mathcal{P}_H^O(x)$ are defined in 
	Lemma \ref{l: priorityintra}  and \ref{l: priorityinter}.
\end{theorem}

\begin{proof}
	The proof is done by sharing the same idea with the proof of Theorem \ref{l: bound2}, thus omitted here.
\end{proof}

Similarly with that in FIFO-order, we present an algorithm to iteratively compute the minimum $m_i$ for each task $\tau_i$ to be schedulable. We start by setting
$m_i=\lceil \frac{\sumC_i-\critical_i}{D_i-\critical_i} \rceil$ for each $\tau_i$ and then gradually increase $m_i$ until finding the minimum value of $m_i$ for $\tau_i$ to be schedulable according to Theorem \ref{l: bound3}. The pseudo-code is shown in Algorithm \ref{al:secondmodelp}.

\begin{algorithm}
	\caption{Processor partitioning algorithm for priority-order.}
	\begin{algorithmic}[1]	
		%\FOR{each task $\tau_i$}
		\STATE For each $\tau_i$: $m_i \gets \lceil \frac{\sumC_i-\critical_i}{D_i-\critical_i} \rceil$;
		\FOR{each task $\tau_i$}
		\WHILE{(1)}
		\FOR{each resource $\resource_q \in \resourceset_i$} 
		%		  \IF {FIFO order is adopted}
		\STATE \label{line:maximump} find $x\in [0, N_{i,q}]$  s.t., $\mathcal{P}^I(x) +\mathcal{P}^O(x)$ is maximal;
		\ENDFOR
		\STATE Compute the WCRT bound $R_i'$ using Theorem \ref{l: bound3};
		\IF{ $R_i' > D_i$}
		\STATE $m_i \gets m_i+1$;
		\ELSE 
		\STATE break;
		\ENDIF
		\ENDWHILE
		\ENDFOR
		\IF{$\sum_{\tau_i\in \tau}m_i>m$}
		\STATE \textbf{return} \textit{unschedulable}
		\ELSE 
		\STATE \textbf{return} \textit{schedulable} 
		\ENDIF
	\end{algorithmic}
	\label{al:secondmodelp}
\end{algorithm}

%The out-most for-loop in Algorithm \ref{al:secondmodelp} iterates for at most $|\Tau|$ times. The while-loop iterates for at most $m$ times.
%For each $x$, the time complexity to compute $\mathcal{P}^I(x) +\mathcal{P}^O(x)$ is $|\Tau|$, 
%so the time complexity to find the desired $x$ in Line \ref{line:maximump} is $O(N_{max} \cdot |\Tau|)$.
%The overall time complexity of Algorithm \ref{al:secondmodelp} is $O(m \cdot |\Tau|^2 \cdot |\Theta| \cdot N_{max})$.
%\todo{do we need to discuss the complexity of algorithm II the same as I?}

%\section{Comparison}\label{s:comparison}
%\guan{to be moved to the experiment section}

	%\input{singleDAG}
	%\input{framework}
	%\input{semifederated}
	%\input{second}
	%\input{example}
	%\input{multiDAG}
	\section{Evaluations}\label{s:evaluate}

In this section, we evaluate the performance of our approaches \todo{in terms of acceptance ratios, i.e., the ratio between the number of task sets that are schedulable and the number of the whole task sets,} in comparison with the state-of-the-art:
\begin{itemize}
	\item XU-U: Algorithm \ref{al:firstmodel} for unordered spin locks.
	\item XU-F: Algorithm \ref{al:secondmodel} for FIFO-ordered spin locks.
	\item XU-P: Algorithm \ref{al:secondmodelp} for priority-ordered spin locks.
	\item SON-F: test for FIFO spin locks in \cite{dinh2018blocking}.
	\item SON-P: test for priority-ordered spin locks in \cite{dinh2018blocking}.
\end{itemize}

In particular, we adopt an optimal priority assignment when evaluating XU-P and SON-P for priority-order, where we try all permutations of priorities for each task set until either the task set is schedulable or all permutations have been checked \footnote{Note that enumerating all possible priority permutations may result in computation explosion when the number of tasks is large (we have at most 10 tasks in a task set in our experiments, i.e., in Figure \ref{fig:sythetic}.(e)). However, proposing methods of priority assignment is out of the scope in this paper. We choose this method only for comparing with the results from \cite{dinh2018blocking} where the optimal priority assignment is shown to have the best performance.}. \todo{There are several different methods to make the priority assignment, such as assign the locking-priorities based on the tasks' relative deadlines or simulated annealing to find an
	approximately optimal priority assignment \cite{dinh2018blocking}. However, we adopt the optimal priority assignment method to make a fair comparison with \cite{dinh2018blocking}, which is also shown with the best performance in \cite{dinh2018blocking}.}

We compare the above approaches with both synthetic workload and workload generated according to realistic OpenMP programs. \todo{It is necessary to mention that we do not make simulations of scheduling or actually execute any programs but test the schedulability of task sets (either synthetic workload or realistic OpenMP programs) by using their parameters according to different approaches listed above.}

\subsection{Synthetic Workload} \label{ss: sythetic}

\begin{figure*}[htb]
	\centering
	\subfigure[Under different $U_{norm}$.]{\includegraphics[width=2in]{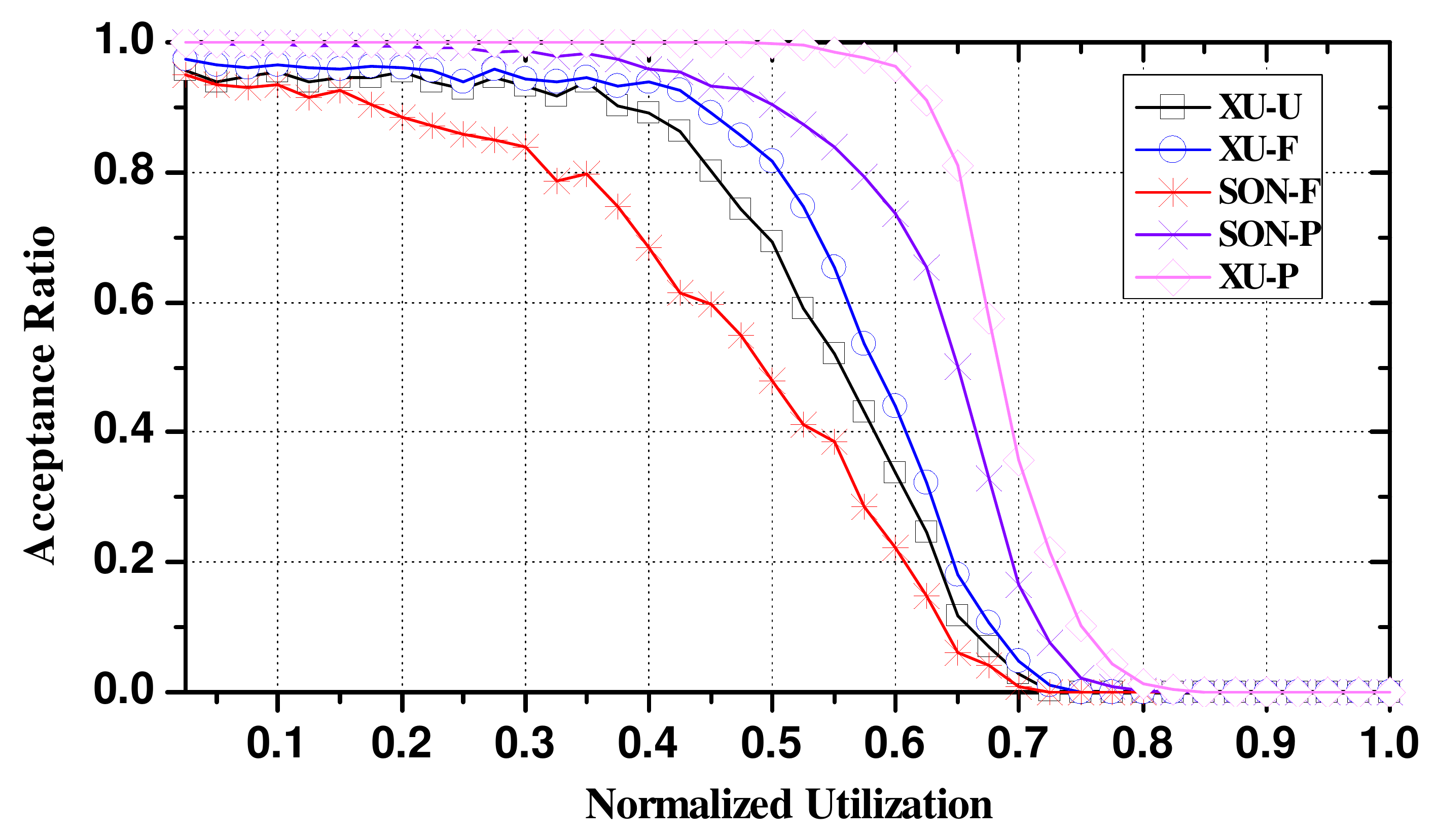}}
%	\hspace{0.15in}
	\subfigure[Under different $\sum_{\tau_i \in \tau}N_{i,q}$.]{\includegraphics[width=2in]{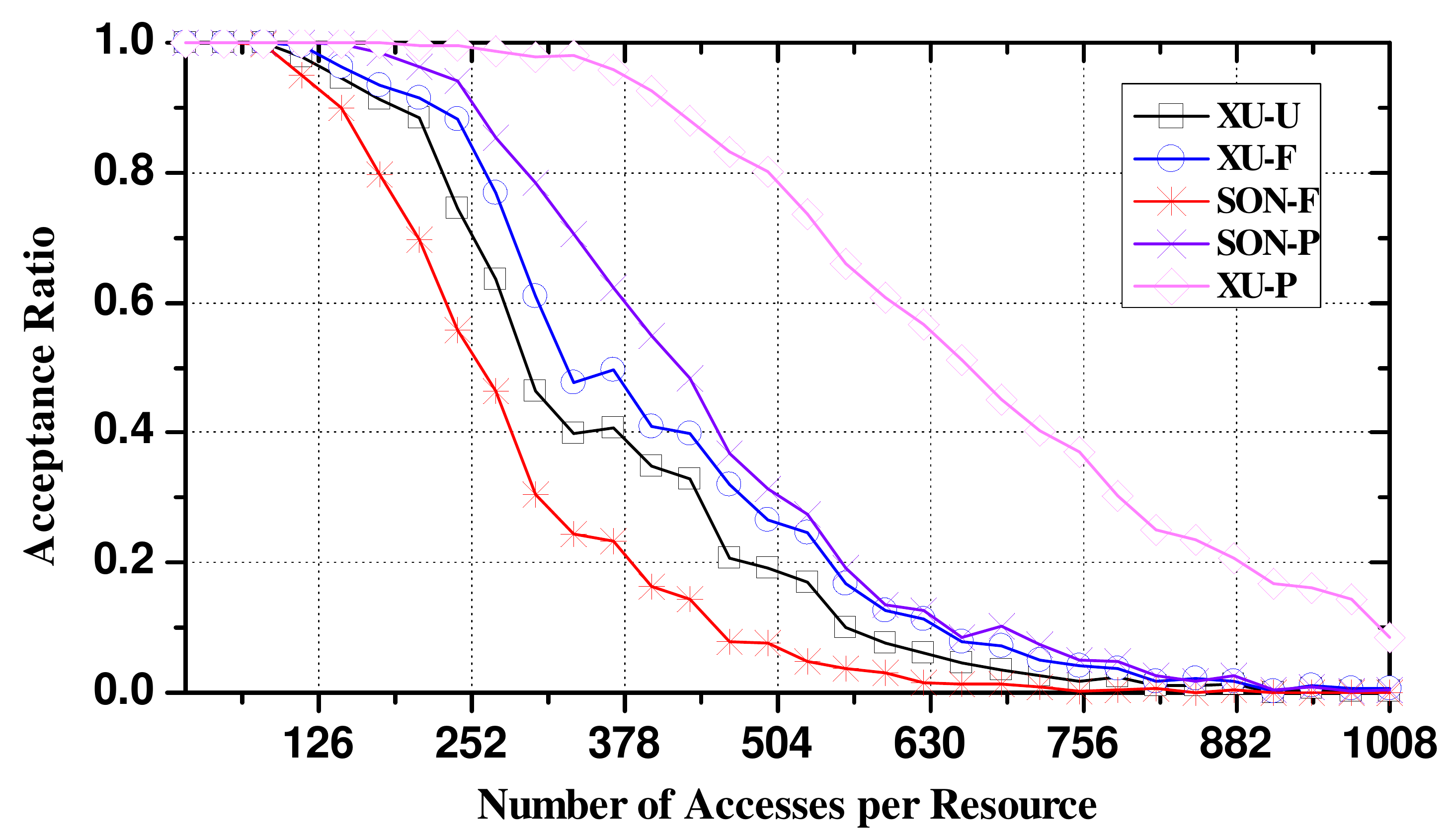}}
%		\hspace{0.15in}
	\subfigure[Under different $|\Theta|$.]{\includegraphics[width=2in]{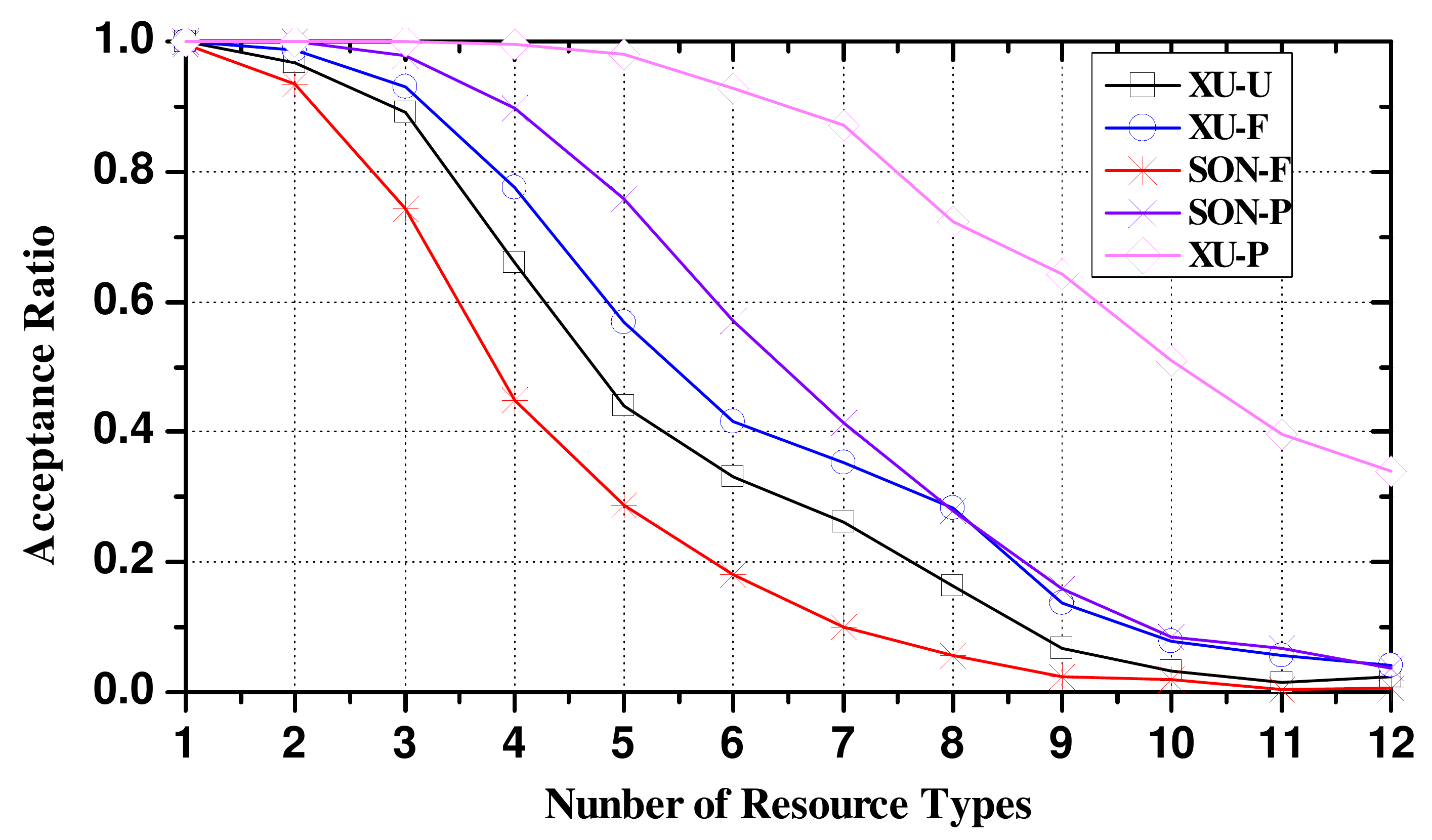}}
		\reduceblack
	\subfigure[Under different $\max_{\forall \tau_i}\{L_{i,q}\}$.]{\includegraphics[width=2in]{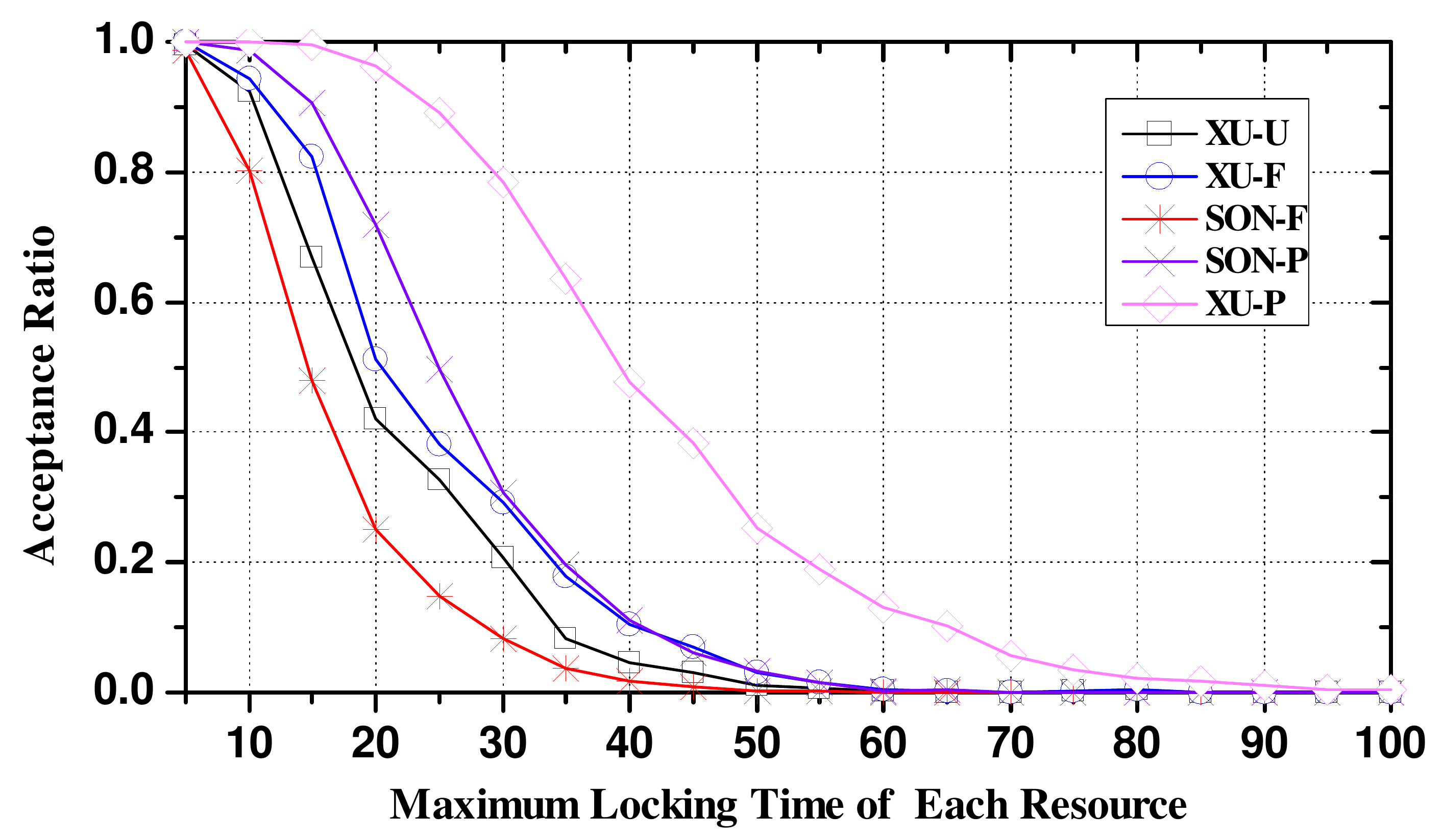}}
%	\hspace{0.15in}
	\subfigure[Under different $n$.]{\includegraphics[width=2in]{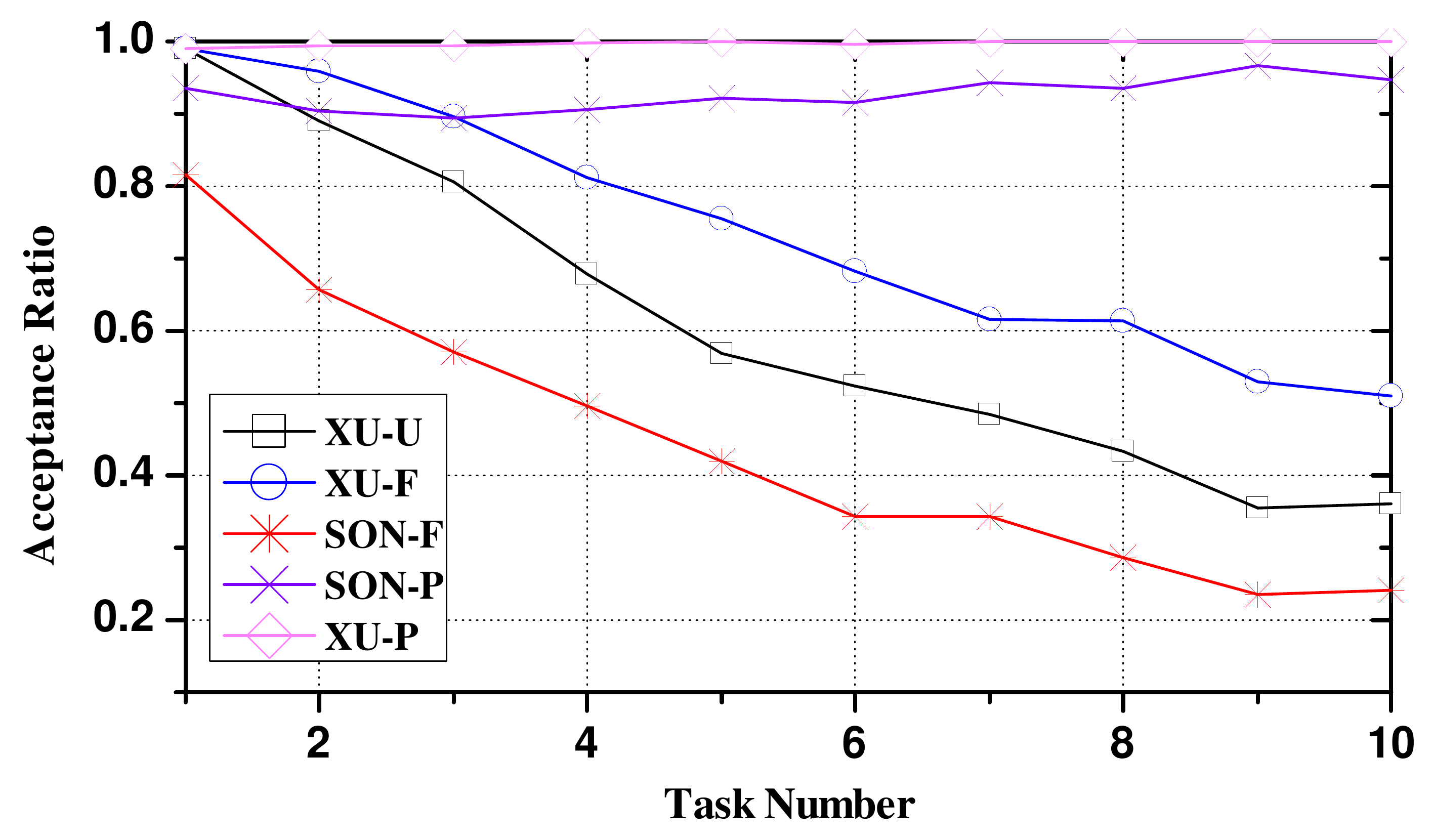}}
%	\hspace{0.15in}
	\subfigure[Under realisitic OpenMP programs.]{\includegraphics[width=2in]{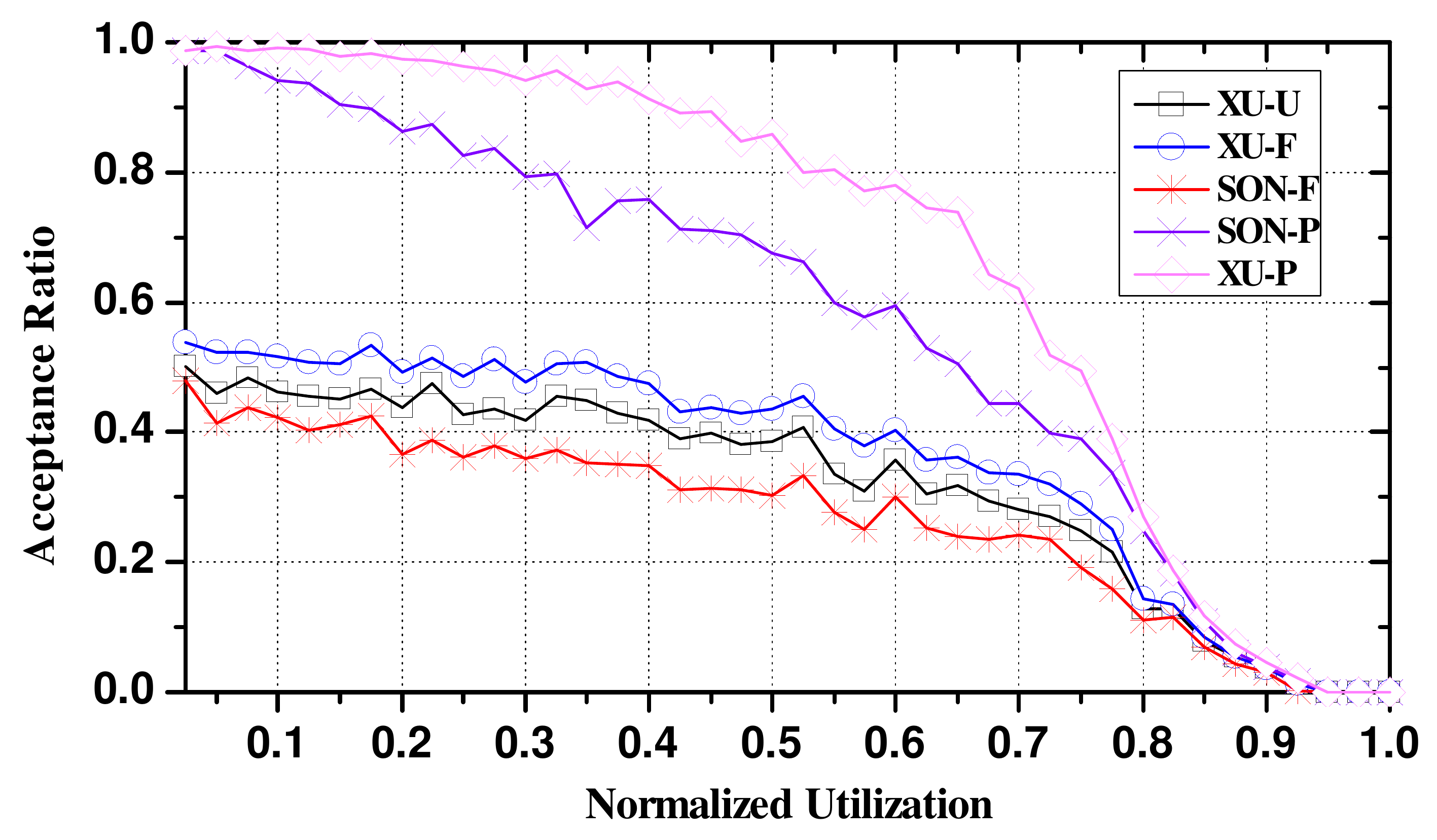}}
%	\hspace{0.15in}
	\caption{Comparisons with the state-of-the-art.}
	\label{fig:sythetic}
\end{figure*}

%\begin{figure*}[htb]
%	\centering
%	\subfigure[Under different normalized utilizations.]{\includegraphics[width=2.2in]{experimentalresults/U2.pdf}}
%	\hspace{0.15in}
%	\subfigure[Under different numbers of accesses per resource.]{\includegraphics[width=2.2in]{experimentalresults/Rnumber2.pdf}}
%		\hspace{0.15in}
%	\subfigure[Under different numbers of resource types.]{\includegraphics[width=2.2in]{experimentalresults/Rtype2.pdf}}\\
%	\subfigure[Under different maxiumum locking times of each resource.]{\includegraphics[width=2.2in]{experimentalresults/Rlength2.pdf}}
%	\hspace{0.15in}
%	\subfigure[Under diffrent numbers of tasks.]{\includegraphics[width=2.2in]{experimentalresults/Tnumber2.pdf}}
%	\hspace{0.15in}
%	\subfigure[Under different $\alpha$.]{\includegraphics[width=2.2in]{experimentalresults/Jitter2.pdf}}
%	\hspace{0.15in}
%	\reduceblack
%	\caption{Comparisons under different dimensions (the Y-axis is the normalized acceptance ratio).}
%	\reduceblack
%	\label{fig:e1}
%\end{figure*}

We first compare the three approaches with randomly generated task systems. 
The DAG tasks are generated as follows:

\begin{itemize}
	\item \textbf{Task Graph $G_i=\langle V_i, E_i \rangle $:} The task graph of each task is generated using the Erd\"{o}s-R\'{e}nyi method $G(|V_i| , p)$ \cite{cordeiro2010random}. For each task, the number of vertices $|V_i|$ is randomly chosen in $[100, 400]$. The WCET of each vertex is randomly picked in  $[250, 600]$. The metrics of the number and WCETs of vertices are consistent with the measurement results in \cite{wang2017benchmarking}. For each possible edge we generate a random value in  $[0, 1]$ and add the edge to the graph only if the generated value is less than a predefined threshold $p=0.1$. The same as in \cite{saifullah2014parallel}, a minimum number of additional edges are added to make a task graph weakly connected.

	\item \textbf{Deadline and Period:} The deadline $D_i$ of each task $\tau_i$ is generated in a similar way with \cite{dinh2018blocking}: after $\critical_i$ is fixed, $D_i$ is generated according to a ratio between $\critical_i$ and $D_i$ randomly chosen in $\{0.125,0.25\}$. The period $T_i$ is set to be equal to $D_i$. 
	
	\item \textbf{Resource}:
The number of resource types is in the range $[1,12]$.
The number of accesses to each resource by all tasks $\sum_{\tau_i \in \tau}N_{i,q}$ is in the range $[16,1008]$, and is randomly distributed to different tasks.
The maximal locking time $\max_{\forall \tau_i}\{L_{i,q}\}$ of each resource is in the range $[5, 60]$ and each $L_{i,q}$ is randomly picked in $[1, \max_{\forall \tau_i}\{L_{i,q}\}]$.
\end{itemize}
%We first assume $Q_{i,q}=N_{i,q}L_{i,q}$, with which 
%\textbf{Model I} is strictly more abstract than \textbf{Model II}, to have a direct comparison of the three approaches without involving the difference between the two models. Later we consider the general case of $Q_{i,q} \leq N_{i,q}L_{i,q}$ and discuss the impact of different models to the analysis results.	
%	\item \textbf{Number of resource types}:
%	\item \textbf{Number of requests per resource}: We vary the number of requests per resource $\sum_{\tau_i \in \tau}N_{i,q}$ in $[128,1024]$. 
%	\item \textbf{maximum length of each request}: We vary the maximum length $\max\{L_{i,q}\}$of each kind of resource in $[5\mu s, 60\mu s]$ and randomly pick $L_{i,q}$ in $[1\mu s, \max\{L_{i,q}\}]$.  

%
%The resources are generated as follows:
%\begin{itemize}
%	\item \textbf{Number of resource types}: We vary the number of resource types in $[1,12]$.
%	\item \textbf{Number of requests per resource}: We vary the number of requests per resource $\sum_{\tau_i \in \tau}N_{i,q}$ in $[128,1024]$. 
%	\item \textbf{maximum length of each request}: We vary the maximum length $\max\{L_{i,q}\}$of each kind of resource in $[5\mu s, 60\mu s]$ and randomly pick $L_{i,q}$ in $[1\mu s, \max\{L_{i,q}\}]$.  
%\end{itemize}

Since we only focus on heavy tasks, a task with $U_i<1$ is discarded until a heavy task is generated during the generation of each task. For each task set, we generate $n$ tasks where $n$ is in $[1,14]$. The normalized utilization $U_{norm}$ (the ratio between the total utilization and the number of processors) of each task set is predefined, which will be explained in detail for the configuration of each figure. After we generate all tasks in a task set, we can compute the total utilization $U_{\sum}$, then we set the number of processors according to the formula $m=\lceil \frac{U_{\sum}}{U_{norm}} \rceil$. \todo{The number of processors could become quite large (far more than 10 processors) when $U_{\sum}$ is relatively low (e.g., lower than 0.2) or the number of tasks in a task set is large (e.g., more than 6 tasks).} For each configuration (corresponding to one point on the X-axis), we generate $1000$ task sets. 

In Figure \ref{fig:sythetic}.(a)-(e), we set a basic configuration and in each group of experiments vary one parameter while keeping others  unchanged.
The basic configuration is
as follows: $n=4$, $U_{norm}=0.5$, the number of resource types is 4, $\sum_{\tau_i \in \tau}N_{i,q}=256$ and $\max_{\forall \tau_i}\{L_{i,q}\}=15$.

Figure \ref{fig:sythetic}.(a) shows acceptance ratios of all tests under different normalized utilizations (X-axis). Figure \ref{fig:sythetic}.(b) evaluates the acceptance ratios under different $\sum_{\tau_i \in \tau}N_{i,q}$. We can observe that the acceptance ratios of all tests decrease as $\sum_{\tau_i \in \tau}N_{i,q}$ increases. Figure \ref{fig:sythetic}.(c) shows the acceptance ratios under different number of resource types. The acceptance ratios of all tests decrease as the number of resource types increases. In Figure \ref{fig:sythetic}.(d), resources are generated with different $\max_{\forall \tau_i}\{L_{i,q}\}$. The schedulability of all tests decreases as $\max_{\forall \tau_i}\{L_{i,q}\}$ increases. In Figure \ref{fig:sythetic}.(e), we generate different number of tasks in each configuration. The schedulability of XU-U, XU-F and SON-F decreases as the number of tasks increases whereas the schedulability of tests for priority order, i.e., XU-P and SON-P, is hardly affected by the number of tasks. 

%In particular, it can be observed that the performance of XU-U decreases significantly as the number of tasks increase, and becomes close to SON-M2-F. This is because XU-M1 uses less information than the other two tests, and the inter-task blocking time calculated in XU-M1 is more pessimistic. When the number of tasks in a task set is larger, this pessimism is enlarged. 

From the above results we see that tests for priority-order perform better than those for FIFO-order and unordered, and our approaches consistently outperform the state-of-the-art under different parameter settings: XU-P outperforms SON-P and XU-F outperforms SON-F.
In particular, even if XU-U adopts less queue order information, it still consistently outperforms SON-F due to our new analysis techniques which systematically analyze the blocking time that may delay the finishing time of a parallel task and jointly consider the impact of blocking time to both the total workload and the longest path length.

\begin{figure*}[htb]
	\centering
	\subfigure[Under different $U_{norm}$.]{\includegraphics[width=2in]{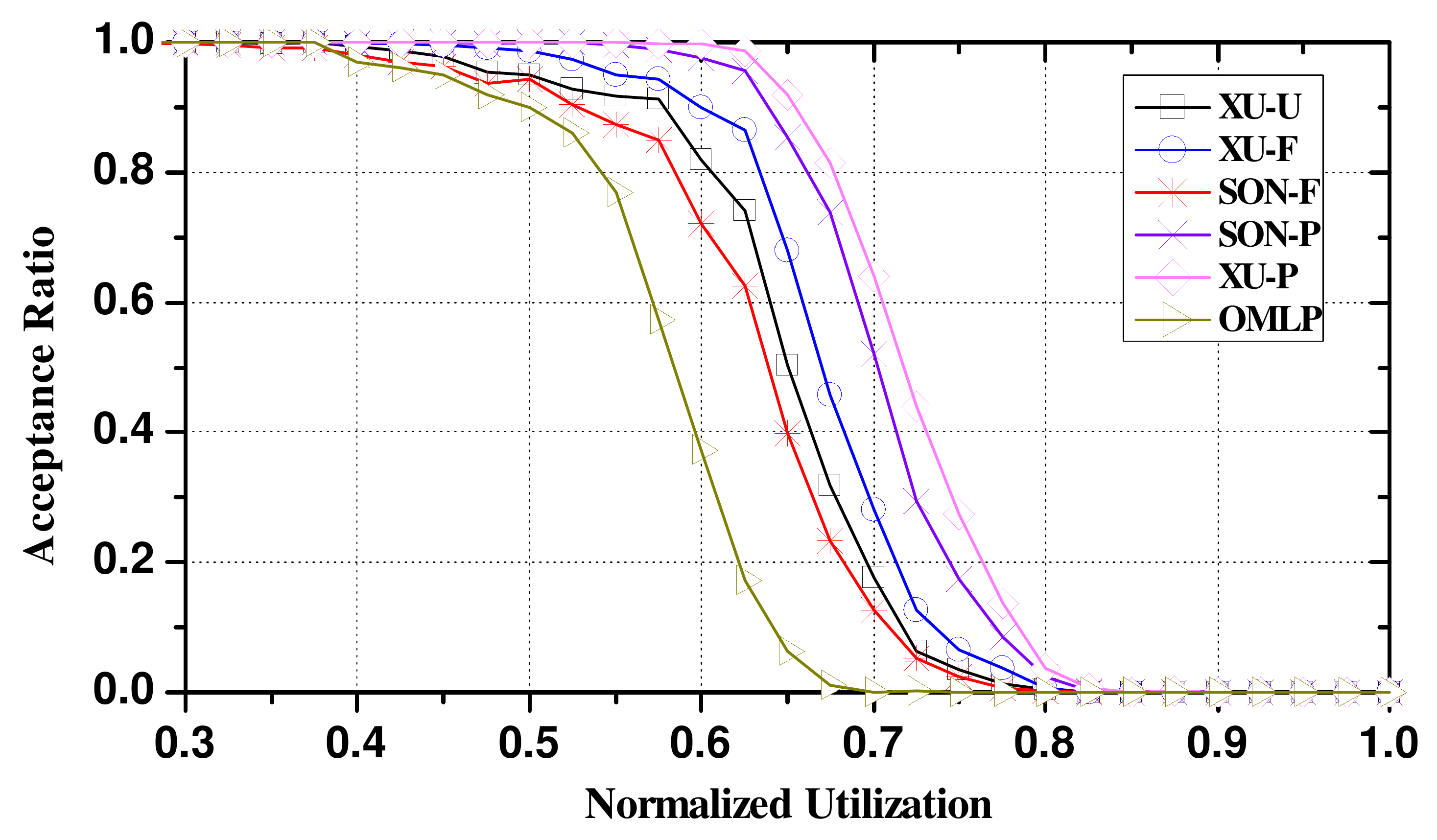}}
	%	\hspace{0.15in}
	\subfigure[Under different $n$.]{\includegraphics[width=2in]{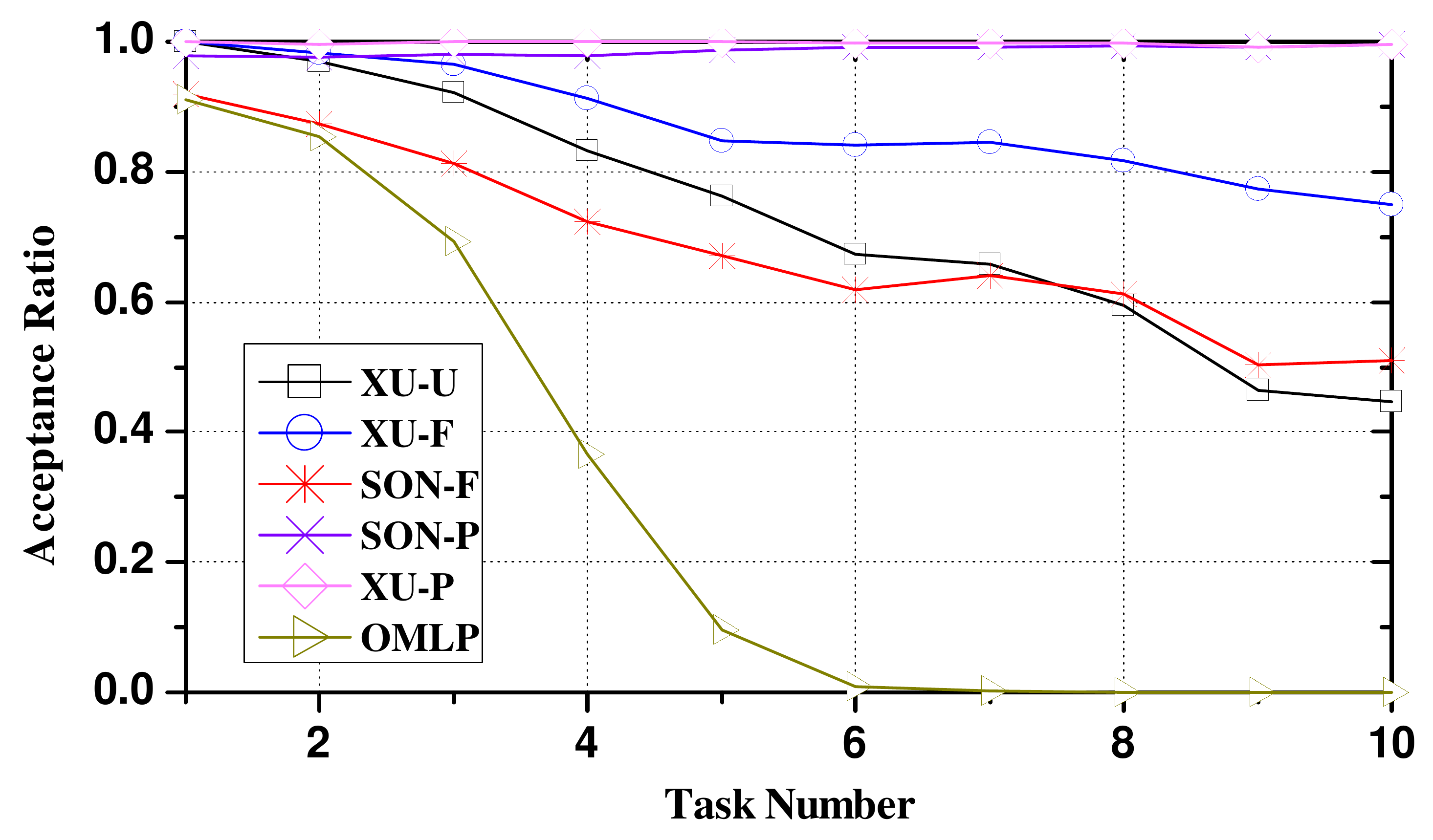}}
	%		\hspace{0.15in}
	\subfigure[Under different $\max_{\forall \tau_i}\{L_{i,q}\}$.]{\includegraphics[width=2in]{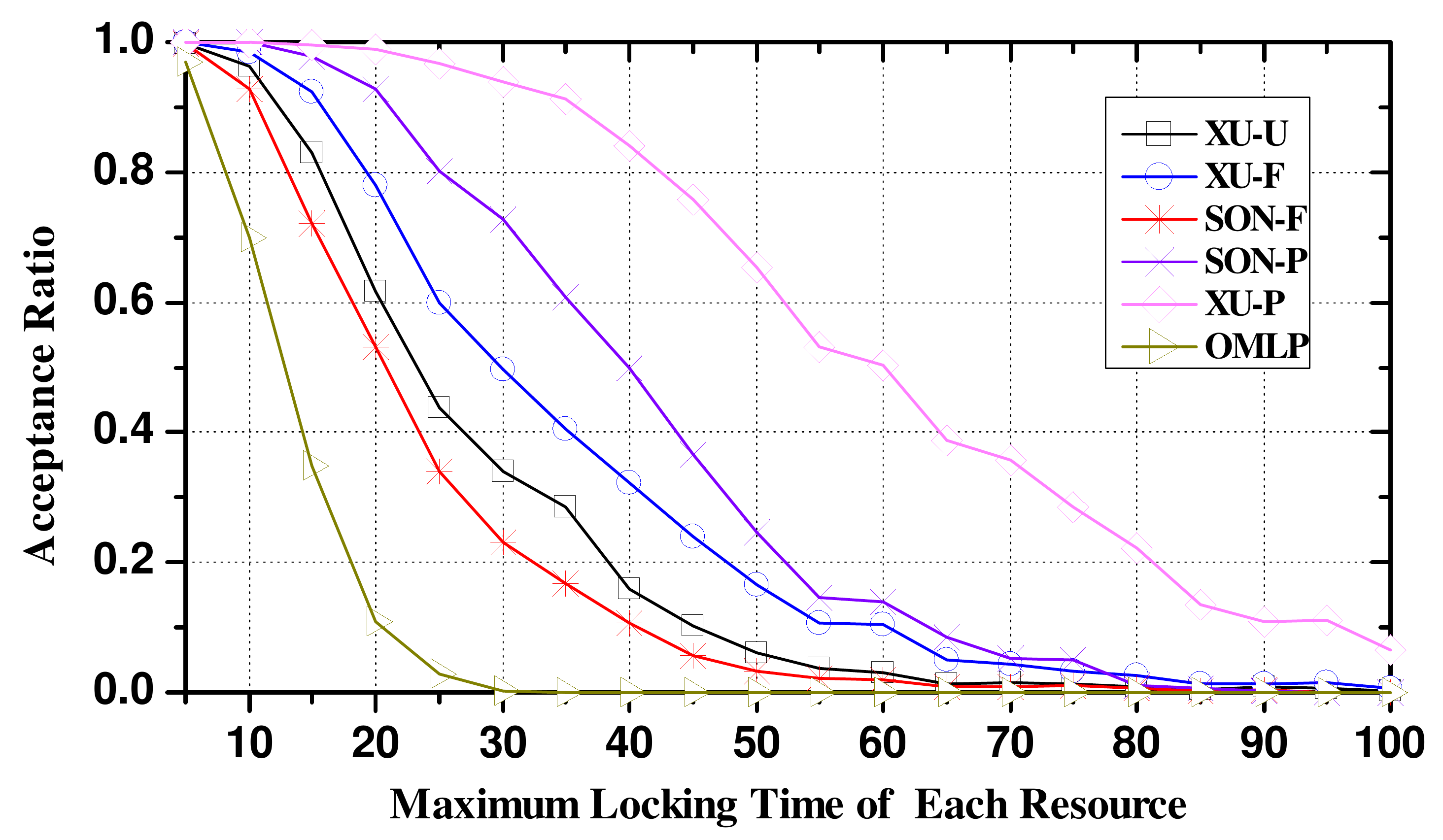}}
	\caption{Comparisons with OMLP.}
	\label{fig:omlp}
\end{figure*}

In the following, we conduct experiments to evaluate the performance of both \cite{dinh2018blocking} and our results in comparison with locking protocols for sequential tasks. That is, we try to find a straightforward way to extend locking protocols for sequential tasks to paralleled tasks, such that the points we make in Section \ref{s: sequentialdiscussion} can be more clear. Some modern analysis techniques for sequential tasks use \emph{Linear Programming} (LP) to achieve more precise performance, e.g., \cite{wieder13rtss,yang2015global}, which are not included due to the following reasons. First, the blocking times are defined under schedulability tests for sequential tasks which can not be directly applied for parallel tasks (some significant modifications and techniques are required and it is not trivial). Second, the LP-based techniques run with significant computing resources and time since they are with quite high complexity (weeks on clustered computers as provided by the authors of \cite{dinh2018blocking}) whereas our tests and \cite{dinh2018blocking} are polynomial. OMLP is a well-known locking protocol of clustered scheduling for sequential tasks \cite{brandenburg2013omlp} which is also the most relevant work with this paper (DAG tasks scheduled under federated scheduling can be regarded as sequential tasks scheduled on clusters). In Fig.\ref{fig:omlp}, we apply OMLP on DAG tasks in a straightforward manner where each vertex in a DAG task is regarded as an independent sequential task. We first randomly distribute the generated requests of each task to its vertices. The priorities of all vertices in a DAG task are set the same as their indexes, and a vertex with a smaller index has a higher priority. We first compute the S-oblivious PI-blocking for each vertex according to the blocking analysis techniques presented in \cite{brandenburg2013omlp} and then add the PI-blocking to the WCET of the vertex, after which the longest length among all paths and the WCET of all vertices of $\tau_i$ are denoted by $\critical_i^{'}$ and $\sumC_i^{'}$ respectively. Then we use the general schedulability test of federated scheduling for each DAG task \cite{li2014analysis}, i.e., the response time of task $\tau_i$ is computed by $R_i \leq \critical_i^{'}+\frac{\sumC_i^{'}-\critical_i^{'}}{m_i}$. The schedualbility of the task set is decided in a similar way with Algorithm \ref{al:secondmodelp} (the only difference is on the computation of $R_i$).

In Figure \ref{fig:omlp}.(a)-(c), we set a basic configuration and in each group of experiments vary one parameter while keeping others unchanged.
The basic configuration is
as follows: $n=4$, $U_{norm}=0.6$, the number of resource types is 9, $\sum_{\tau_i \in \tau}N_{i,q}=60$ and $\max_{\forall \tau_i}\{L_{i,q}\}=15$. In comparison with the basic configuration of Figure \ref{fig:sythetic}, we have significantly reduced the total number of resource accesses to evaluate the performance of our results in a system with a modicum number of accesses (the case that is more close to the practical scenarios). It may be noticed that both our work and \cite{dinh2018blocking} are based on the classic Graham's bound \cite{graham1969bounds}. Thus if there are no resource access contentions, the schedulabilities of our result and \cite{dinh2018blocking} are the same, and of course the gap of the performance between our results and \cite{dinh2018blocking} becomes more significant when there are more resource access contentions. From Figure \ref{fig:omlp}, we can observe that our results still outperform \cite{dinh2018blocking}. Moreover, even more concrete information are used (i.e., the exact distributions of requests), directly applying locking protocols and associated blocking analysis techniques for sequential tasks on DAG tasks is quite pessimistic (as discussed in Section \ref{s: sequentialdiscussion}).

\subsection{Realistic OpenMP Programs}\label{ss: benchmark}
In the following, we evaluate the three approaches with workload generated according to realistic OpenMP programs. OpenMP supports task parallelization since version 3.0 \cite{board2008openmp}, which can be modeled as DAG models \cite{wang2017benchmarking}. We
collect $8$ OpenMP programs (see Table. \ref{tb:benchmark}) using C language from different benchmark suits and transform them into DAG model.
We measure the $\sumC_i$ and $\critical_i$ of each program and $N_{i,q}$ and $L_{i,q}$ to each shared resource by each task on a hardware platform with Intel i7-7820HQ CPU@2.90GHz, cache size of 8MB and total memory of 4GB. The run time  compiling environment is Ubuntu 12.04.5 LTS with gcc 4.9.4. We consider $10$ different types of resources, where the first $4$ are shared data objects in the operating kernel accessed via system calls (i.e., time) or \todo{library calls (i.e., fprintf, printf, malloc)}. The remaining $6$ 
are shared data structures or non-reusable routines protected by \emph{\# pragma omp critical} in the OpenMP program.

The measurement results are summarized in Table \ref{tb:benchmark}, where the time unit is $\mu s$. 
Note that the measurement results are \emph{not} guaranteed to be safe upper bounds of the desired parameters. In order to obtain their safe upper bounds, a comprehensive static analysis covering all the hardware and software behaviors is required. In this paper, we simply use these results to approximately represent the workload characteristics of 
these OpenMP programs. \todo{It may be notices that the number of types of shared resources that each program may access is not large. The resources accessing behaviors of programs in bots-1.1.2 are similar because of that they are all commutative algorithms and may use some similar library calls such as "malloc". These features does not affect our evaluations, and the main purpose of our evaluations is to show the impact to the schedulability of realistic parallel programs with shared resources and the schedulabilities under different methods.}

\begin{table*}[!htb]\caption{Measurement results of OpenMP programs. }
	\centering
	%	\scriptsize
	\footnotesize
	%	\tiny
	\begin{tabular}{ |l@{}|l@{}|l@{}|l@{}|l@{}| l@{}|l@{} | l@{}|l@{} |l@{}|l@{}| l@{}|l@{}| l@{}|l@{}| l@{}|l@{}| l@{}|l@{} |l@{}|l@{} |l@{}|l@{} |l@{}|}
		\hline		
		\multirow{2}{*}{Benchmark} & \multirow{2}{*}{Application} & \multirow{2}{*}{$\sumC_i$}  &\multirow{2}{*}{ $\critical_i$} & \multicolumn{2}{c|}{$\resource_0$}&\multicolumn{2}{c|}{$\resource_1$}&\multicolumn{2}{c|}{$\resource_2$}&\multicolumn{2}{c|}{$\resource_3$}&\multicolumn{2}{c|}{$\resource_4$}&\multicolumn{2}{c|}{$\resource_5$}&\multicolumn{2}{c|}{$\resource_6$}&\multicolumn{2}{c|}{$\resource_7$}&\multicolumn{2}{c|}{$\resource_8$}&\multicolumn{2}{c|}{$\resource_9$} \\
		\cline{5-24}
		&&&& $N$ &$L$ & $N$ &$L$& $N$& $L$& $N$& $L$& $N$& $L$& $N$&$ L$& $N$& $L$& $N$& $L$& $N$& $L$& $N$& $L$\\ \hline
		\multirow{8}{*}{bots-1.1.2 \cite{duran2009barcelona}}
		& alignment.for& 313168 & 11446& 22& 2& 1 &2& 2& 2& 0& 0& 0 &0& 0& 0& 0& 0& 0 &0& 0& 0& 0& 0 \\
		& alignment.single& 315981  & 9980& 22& 2& 1& 2& 2& 2& 0& 0& 0 &0& 0& 0& 0& 0& 0& 0& 0& 0& 0& 0  \\
		& fft & 274 & 58& 21 &2& 1 &4& 2 &2& 0& 0& 0 &0& 0 &0& 0 &0& 0 &0& 0 &0& 0& 0  \\
		& fib & 353  & 20 & 20& 2& 0& 0& 2& 2& 0& 0& 0 &0& 0 &0& 0 &0& 0 &0& 0& 0& 0 &0 \\
		& sort & 1757 & 217 & 20& 2& 2& 4& 2& 2& 0& 0& 0& 0& 0& 0& 0 &0& 0& 0& 0& 0& 0& 0 \\
		%		& sparselu.gcc.for & 202418 & 5968  \\
		%		& sparselu.gcc.single& 206836   & 29481  \\
		& floorplan& 5843&92& 36 &2& 6 &1& 2& 2& 0 &0& 4& 1& 0& 0& 0 &0& 0 &0& 0& 0& 0& 0  \\
		\hline
		\multirow{2}{*}{OpenMPMicro \cite{dimakopoulos2008microbenchmark}} 
		& MatrixMultiplication & 5873246 & 106983 & 0& 0& 3 &7& 0 &0& 5& 4& 0 &0& 0 &0& 0 &0& 0& 0& 0 &0& 0& 0  \\
		& Square & 50000812 & 1000066& 0& 0& 0 &0& 0& 0& 0& 0& 0 &0& 20 &5& 50& 1& 50 &105& 50 &79& 50& 1  \\
		%	& BlackScholes & 21633565 & 46059 \\
		\hline	
	\end{tabular}
	\label{tb:benchmark}
\end{table*}

For each task set, we pick $n$ programs (each being a DAG task) in Table \ref{tb:benchmark}, where $n$ is randomly chosen in $[2,5]$. The deadline $D_i$ of each task 
and the number of processors in the system are set in the same way as Section \ref{ss: sythetic}.
%is generated according to a ratio between $\critical_i$ and $D_i$ and we set the number of processors according to the formula $m=\lceil \frac{U_{\sum}}{U_{norm}} \rceil$, i.e., the same as in \ref{ss: sythetic}. 

Fig. \ref{fig:sythetic}.(f) shows acceptance ratios of all tests under different normalized utilizations (X-axis).
 We can observe that the acceptance ratios of all tests decrease in comparison with Fig. \ref{fig:sythetic}.(a). This is because some applications have relatively short deadlines and periods by our task generation method, and thus have low tolerance 
to blocking time caused by other tasks.
Nevertheless, the results have the same trend as in Fig. \ref{fig:sythetic}.(a): XU-F consistently outperforms SON-F while XU-P outperforms SON-P. 

%\begin{figure}[!htb]
%	\centering
%	\includegraphics[width=2.73in]{experimentalresults/Rbench.pdf}
%	\reduceblack
%	\caption{Comparison under realistic openMP programs.}
%	\reduceblack
%	\label{fig:benchmark}
%\end{figure}

	\section{Related Work} \label{s:related}

There is plentiful of literature on scheduling algorithms and analysis techniques for the parallel real time tasks \cite{li2014analysis,maia2014response,jiang2016decomposition,fonseca2017improved,jiang2017semi}, which all assume tasks to be independent from each other and do not consider the locking issue. 
%Under global scheduling, a capacity augmentation bounds of $\frac{3+\sqrt{5}}{2}$ for GEDF has been provided in \cite{li2014analysis} which is proved to be tight. In \cite{saifullah2014parallel}, a decomposition based method was proposed which was then refined and improved in \cite{jiang2016decomposition}. In \cite{maia2014response}, a response time analysis was presented for scheduling DAG tasks with conditional branches. Under federated scheduling, an optimal capacity augmentation bound of $2$ was proved in \cite{li2014analysis}. Baruah then applied federated scheduling on DAG tasks with constrained deadline \cite{baruah2015federatedconstraind} and arbitrary deadline \cite{baruah2015federated}. Brandenburg et al. \cite{brandenburg2010optimality} classified the analysis techniques into two types: the suspension-oblivious and the suspension-aware analysis, and also gave an asymptotically optimal metric to evaluate the performance of real-time locking protocols.  

Real-time locking protocols are well supported in uniprocessor systems. The Priority Inheritance Protocol (PIP) \cite{sha1990priority} is the first solution to address the priority inversion problem. There are several optimal protocols for uniprocessor real-time task systems, such as Multiprocessor Stack Resource Policy (SRP) \cite{baker1991stack} and Priority Ceiling Protocol (PCP) \cite{sha1990priority} which guarantee bounded blocking time for a single resource access request and ensure deadlock freedom. 
 
On multiprocessors, there are two major lock types: spin locks and suspension-based semaphores. Much
work has been done for partitioned multiprocessor scheduling, such as MPCP \cite{rajkumar1990real} and  DPCP \cite{rajkumar1988real} and the Multiprocessor Stack Resource Policy (MSRP) \cite{gai2001minimizing}. The Flexible Multiprocessor Locking Protocol (FMLP) \cite{block2007flexible} is a family of locking protocols which support both global and partitioned scheduling. The Parallel Priority Ceiling Protocol (P-PCP) \cite{easwaran2009resource} is an extension of the PIP that attempts to avoid certain unfavorable
blocking situations.  The family of $O(m)$ Locking Protocols (OMLP) \cite{brandenburg2010optimality, brandenburg2013omlp} is a suite of suspension-based locking protocols that have proved to be asymptotically optimal under suspension-oblivious analysis.  Lakshmanan et al. \cite{lakshmanan2009coordinated} proposed the Multiprocessor Priority Ceiling Protocol with virtual spinning and Faggioli et al. \cite{faggioli2010multiprocessor} proposed a locking protocol for reservation-based schedulers that includes preemptable spinning.
 
A recent work considering locks for parallel real-time task model is \cite{dinh2018blocking}, which adopts the federated scheduling framework with spin locks. As mentioned before, \cite{dinh2018blocking} analyzes the impact of the blocking time to the total workload and the longest path length separately, which leads to significant pessimism in analysis precision. The contribution of this paper is to address this 
 pessimism in \cite{dinh2018blocking}.
 
 The locking protocols have been studied with other graph-based task models, such as the DRT model \cite{guan2011resource} and multi-frame task model\cite{ekberg2015optimal}. However, 
 these models are still sequential (multiple edges going out from a vertex have conditional branching semantics rather than forking).
 
 \reduceblack
  \reduceblackhalf
%  DRT task model has branching 
% 
% in  \cite{}
% 
% 
%   
%% However, their work analyze the maximum overall blocking and the critical path blocking separately which may introduce much pessimistic. To the best of our knowledge, this is the only work of the state-of-the-art analyzing blocking for locks in parallel real time tasks.   
%
%\guan{add one paragraph one locking for DRT}

	\section{CONCLUSIONS} \label{s:conclusion}
We study the analysis of parallel real-time tasks with spin locks in three different orders under federated scheduling. A recent work \cite{dinh2018blocking} developed analysis techniques for this problem, which are pessimistic since all blocking time are assumed to delay the finishing time of a parallel task and the blocking time to the total workload and the longest path length of each task is analyzed separately. In this paper, we develop new schedulability and blocking analysis techniques to improve the analysis precision. In our future work, we will investigate blocking analysis on other (finer-grained) models.

%A recent work \cite{dinh2018blocking} developed analysis techniques for this problem, which are pessimistic since the blocking time to the total workload and the longest path length of each task is analyzed separately. In this paper, 

	% Can use something like this to put references on a page
	% by themselves when using endfloat and the captionsoff option.
	\ifCLASSOPTIONcaptionsoff
	\newpage
	\fi
	
	%\input{segmentation_proof}
	
	% trigger a \newpage just before the given reference
	% number - used to balance the columns on the last page
	% adjust value as needed - may need to be readjusted if
	% the document is modified later
	%\IEEEtriggeratref{8}
	% The "triggered" command can be changed if desired:
	%\IEEEtriggercmd{\enlargethispage{-5in}}
	
	% references section
	
	% can use a bibliography generated by BibTeX as a .bbl file
	% BibTeX documentation can be easily obtained at:
	% http://www.ctan.org/tex-archive/biblio/bibtex/contrib/doc/
	% The IEEEtran BibTeX style support page is at:
	% http://www.michaelshell.org/tex/ieeetran/bibtex/
	%\bibliographystyle{IEEEtran}
	% argument is your BibTeX string definitions and bibliography database(s)
	%\bibliography{IEEEabrv,../bib/paper}
	%
	% <OR> manually copy in the resultant .bbl file
	% set second argument of \begin to the number of references
	% (used to reserve space for the reference number labels box)
	\bibliographystyle{IEEEtran}
	% argument is your BibTeX string definitions and bibliography database(s)
	\bibliography{IEEEabrv,ADA1}

	\begin{IEEEbiography}[{\includegraphics[width=1in,height=1.25in,clip,keepaspectratio]
			{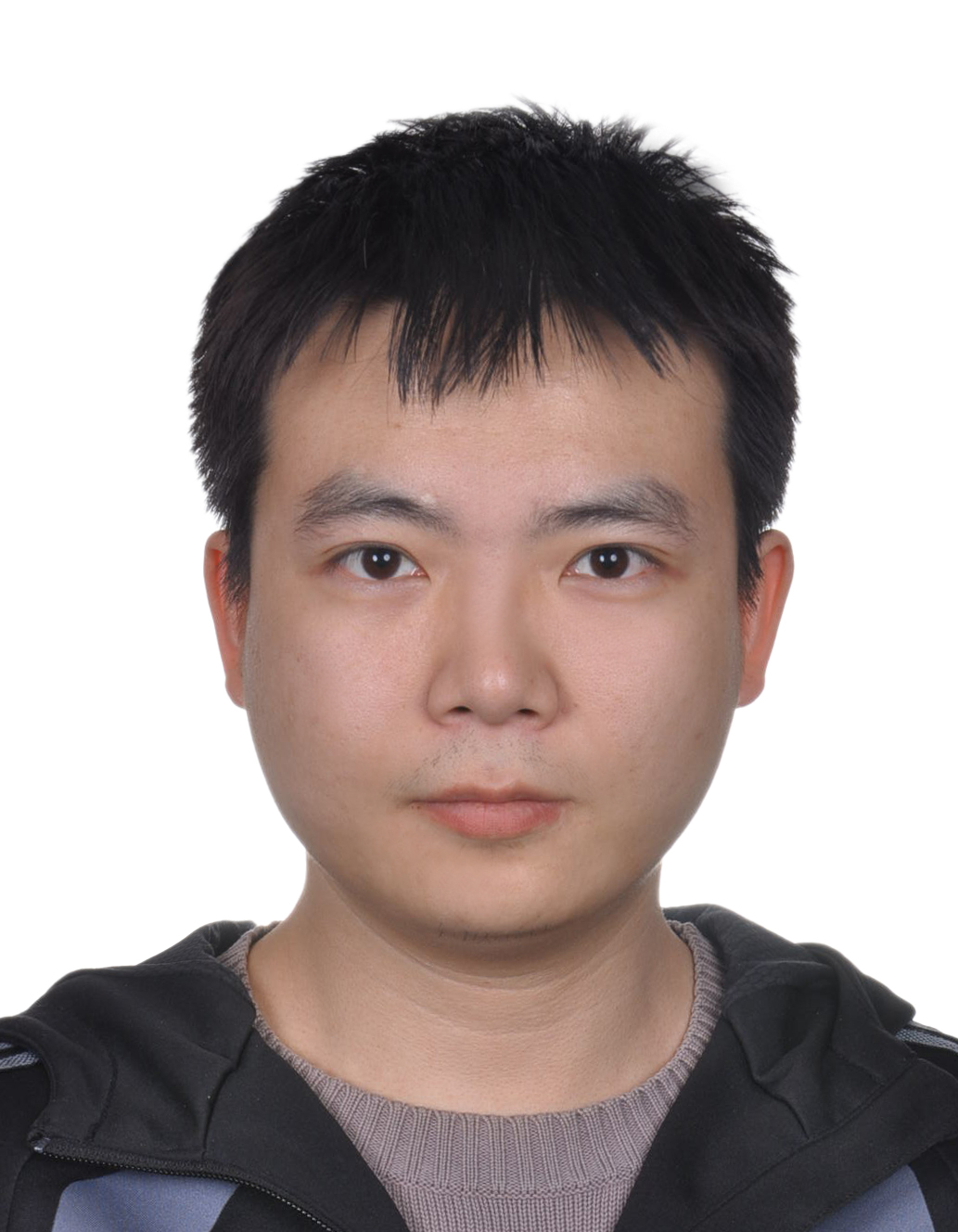}}]{Xu Jiang} has received his BS degree in
		computer science from Northwestern Polytechnical University, China in 2009, received the MS
	degree in computer architecture from Graduate School of the Second Research Institute of
China Aerospace Science and Industry Corporation, China in 2012, and PhD from Beihang University, China in 2018. Currently, he is working in Northeastern University, China. His research interests include real-time systems, parallel and distributed
systems and embedded systems.
	\end{IEEEbiography}
	
	\begin{IEEEbiography}[{\includegraphics[width=1in,height=1.25in,clip,keepaspectratio]
			{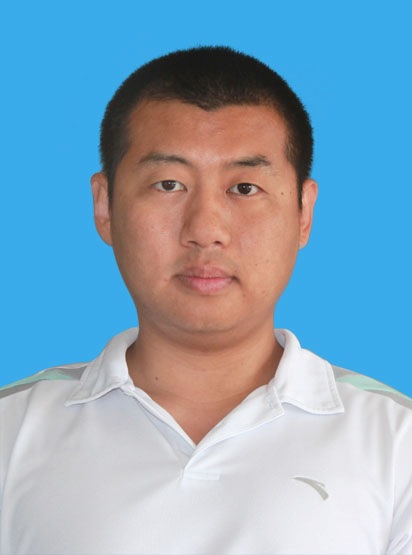}}]{Nan Guan}
		is currently an assistant professor at the Department of Computing, The Hong Kong Polytechnic University. Dr Guan received his BE and MS from Northeastern University, China in 2003 and 2006 respectively, and a PhD from Uppsala University, Sweden in 2013. Before joining PolyU in 2015, he worked as a faculty member in Northeastern University, China. His research interests include real-time embedded systems and cyber-physical systems. He received the EDAA Outstanding Dissertation Award in 2014, the Best Paper Award of IEEE Real-time Systems Symposium (RTSS) in 2009, the Best Paper Award of Conference on Design Automation and Test in Europe (DATE) in 2013.	
	\end{IEEEbiography}
	
	\begin{IEEEbiography}[{\includegraphics[width=1in,height=1.25in,clip,keepaspectratio]
			{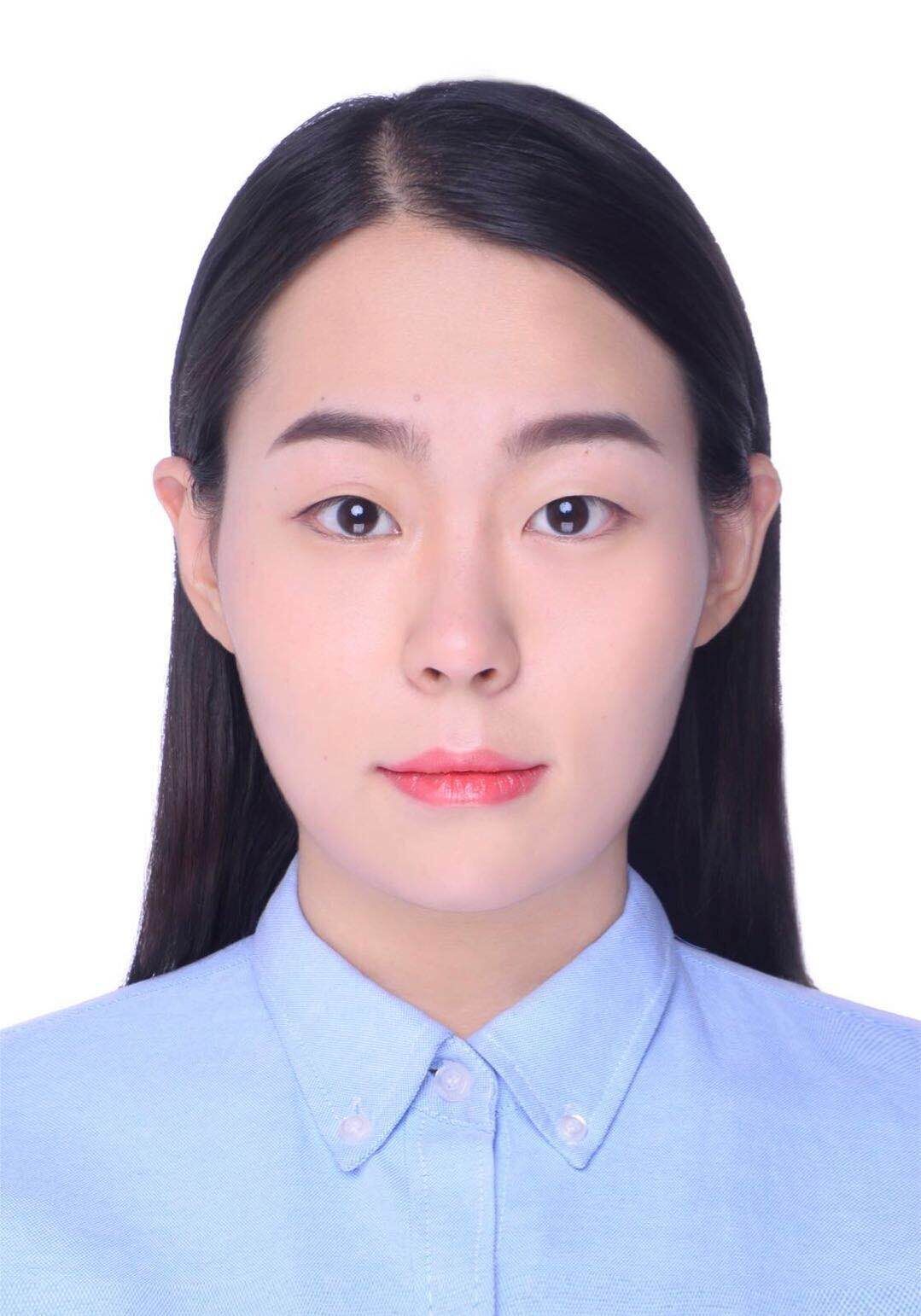}}]{He Du}
		is currently a Ph.D. candidate at School of Computer Science and Engineering, Northeastern University. She received the Bachelor degree from Northeastern University, Shenyang, China, in 2015. Her research interests focus on parallelism program analyze and multiprocessor real-time scheduling.	
	\end{IEEEbiography}

	\begin{IEEEbiography}[{\includegraphics[width=1in,height=1.25in,clip,keepaspectratio]
			{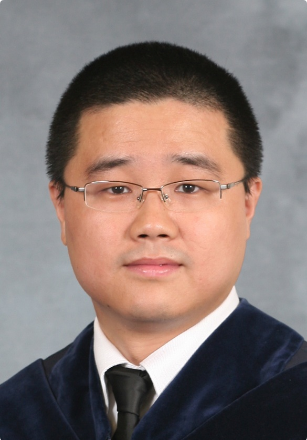}}]{Weichen Liu} received the B.Eng.
		and M.Eng. degrees from the Harbin Institute of
		Technology, Harbin, China, and the Ph.D. degree
		from the Hong Kong University of Science and
		Technology, Hong Kong.
		He is an Assistant Professor with the School
		of Computer Science and Engineering, Nanyang
		Technological University, Singapore. He has
		authored and co-authored over 70 publications in
		peer-reviewed journals, conferences, and books. His
		current research interests include embedded and
		real-time systems, multiprocessor systems, and network-on-chip.
		Dr. Liu was a recipient of the Best Paper Candidate Awards from
		ASP-DAC 2016, CASES 2015, and CODES+ISSS 2009, the Best Poster
		Awards from RTCSA 2017 and AMD-TFE 2010, and the most popular
		Poster Award from ASP-DAC 2017.
	\end{IEEEbiography}
	
	\begin{IEEEbiography}[{\includegraphics[width=1in,height=1.25in,clip,keepaspectratio]
			{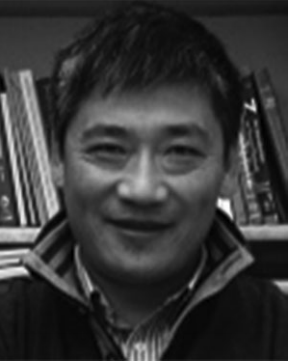}}]{Wang Yi} received the PhD in computer science
		from Chalmers University of Technology, Sweden,
		in 1991. He is a chair professor with Uppsala
		University. His interests include models, algorithms
		and software tools for building and analyzing
		computer systems in a systematic manner to
		ensure predictable behaviors. He was awarded
		with the CAV 2013 Award for contributions to
		model checking of real-time systems, in particular
		the development of UPPAAL, the foremost tool
		suite for automated analysis and verification of
		real-time systems. For contributions to real-time systems, he received
		Best Paper Awards of RTSS 2015, ECRTS 2015, DATE 2013 and
		RTSS 2009, Outstanding Paper Award of ECRTS 2012 and Best Tool
		Paper Award of ETAPS 2002. He is on the steering committee of
		ESWEEK, the annual joint event for major conferences in embedded
		systems areas. He is also on the steering committees of ACM EMSOFT
		(co-chair), ACM LCTES, and FORMATS. He serves frequently on Technical
		Program Committees for a large number of conferences, and was
		the TPC chair of TACAS 2001, FORMATS 2005, EMSOFT 2006, HSCC
		2011, LCTES 2012 and track/topic Chair for RTSS 2008 and DATE
		2012-2014. He is a member of Academy of Europe (Section of Informatics)
		and a fellow of the IEEE.
	\end{IEEEbiography}

	% You can push biographies down or up by placing
	% a \vfill before or after them. The appropriate
	% use of \vfill depends on what kind of text is
	% on the last page and whether or not the columns
	% are being equalized.
	
	%\vfill
	
	% Can be used to pull up biographies so that the bottom of the last one
	% is flush with the other column.
	%\enlargethispage{-5in}
	
	% that's all folks
\end{document}